\documentclass[journal,12pt,onecolumn]{IEEEtran}

\usepackage{amsmath,amsfonts}
\usepackage{amssymb,verbatim,amsfonts,amsmath,amsthm}
\usepackage{geometry}
\usepackage{cite}
\usepackage{hyperref}
\usepackage{bm}

\usepackage{longtable}
\usepackage{booktabs}
\usepackage[justification=centering]{caption}
\usepackage{graphicx}
\usepackage[figuresright]{rotating}
\usepackage{mathtools}
\usepackage{makecell}
\usepackage{seqsplit}

\newtheorem{theorem}{Theorem}
\newtheorem{lemma}[theorem]{Lemma}
\newtheorem{corollary}[theorem]{Corollary}
\newtheorem{proposition}[theorem]{Proposition}
\newtheorem{example}[theorem]{Example}
\newtheorem{remark}[theorem]{Remark}
\newtheorem{definition}[theorem]{Definition}

\DeclarePairedDelimiter\ceil{\lceil}{\rceil}
\DeclarePairedDelimiter\floor{\lfloor}{\rfloor}
\def\lmd{{\rm lmd}}

\begin{document}
\title{Non-special Divisors, LCPs of Codes, and LCD Codes on Kummer Extensions}
\author{Huachao Zhang\thanks{Huachao Zhang is with the School of Mathematics, Sun Yat-sen University, Guangzhou 510275, China (e-mail: zhanghch56@mail2.sysu.edu.cn).},
Chang-An Zhao$^*$\thanks{Chang-An Zhao is with the School of Mathematics, Sun Yat-sen University, Guangzhou 510275, China, and also with the Guangdong Provincial Key Laboratory of Information Security Technology, Guangzhou 510006, China (e-mail: zhaochan3@mail.sysu.edu.cn).\\
*-corresponding author}}
\maketitle

\begin{abstract}
Recently, constructions of linear complementary pairs (LCPs) of codes and linear complementary dual (LCD) codes on function fields have attracted considerable attention due to the wide range of applications of these codes.
Such constructions rely on non-special divisors of degrees $g$ and $g-1$.
In this work, we investigate Kummer extensions defined by $y^m = f(x)$ with $f(x)\in\mathbb{F}_q(x)$ and establish an arithmetic characterization of non-special divisors whose support can contain non-totally ramified places.
Based on this characterization, we explicitly construct non-special divisors of degree $g-1$ on the GK curve.
Moreover, utilizing pure gaps, we explicitly provide several families of effective non-special divisors of degree $g$ on Kummer extensions with the same multiplicities.
We then develop a general framework for constructing LCPs of algebraic geometry (AG) codes on Kummer extensions. 
By virtue of canonical divisors, we show that the security parameters of LCPs of AG codes can be determined within this framework, which also enables the construction of LCD AG codes.
Finally, we illustrate our results with representative examples, including LCPs of codes on the GK curve and LCD codes on quotients of the Hermitian curve.
\,

\end{abstract}

\section{Introduction}
In recent years, linear complementary pairs (LCPs) of codes and linear complementary dual (LCD) codes have become increasingly popular research topics owing to their broad applications in coding theory and cryptography.
These applications include protection against side-channel attacks (SCA) and fault injection attacks (FIA) \cite{bringerOrthogonalDirect2014,carletComplementaryDual2015}, as well as the strengthening of encoded circuits against hardware Trojans \cite{ngoEncodingState2014,ngoLinearComplementary2015}.

LCD codes were first proposed by Massey in \cite{masseyLinearCodes1992}. Later, Yang and Massey gave the necessary and sufficient condition for a cyclic code to be an LCD code \cite{yangConditionCyclic1994}.
In 2018, Carlet et al. \cite{carletLinearCodes2018} proved that a linear code over $\mathbb{F}_q$ is equivalent to an LCD code for $q>3$. 
In the same year, Jin and Xing \cite{jinAlgebraicGeometry2018} showed that an algebraic geometry code over a finite field of even characteristic is equivalent to an LCD code.
LCPs of codes, a generalization of LCD codes, were first introduced by Ngo et al. in \cite{ngoLinearComplementary2015}. In 2018, Carlet et al. \cite{carletLinearComplementary2018} obtained characterizations for LCPs of constacyclic codes and LCPs of quasi-cyclic codes.
Then, G\"uneri, \"Ozkaya and Say{\i}c{\i} \cite{guneriLinearComplementary2018} extended this result to an $n$D cyclic LCP of codes.
Additional notable contributions concerning LCD codes and LCPs of codes can be found in \cite{sendrierLinearCodes2004,carletEuclideanHermitian2018,sokConstructionsOptimal2018,bhowmickLinearComplementary2024a,liSeveralConstructions2024,borelloNoteLinear2020,huLinearComplementary2021,ishizukaConstructionBoth2023,haradaConstructionBinary2021} and the references cited therein.

Since Goppa introduced algebraic geometry (AG) codes on function fields over finite fields \cite{goppaALGEBRAICOGEOMETRICCODES1983}, research on such codes has attracted tremendous attention over the past few decades.
Accordingly, the characterization and construction of LCD AG codes and LCPs of AG codes on function fields have also become popular research directions.
In \cite{beelenExplicitMDS2018}, Beelen and Jin gave an explicit construction of several classes of LCD MDS codes on rational function fields.
Mesnager, Tang, and Qi \cite{mesnagerComplementaryDual2018} provided a construction scheme for obtaining LCD AG codes, and derived some explicit LCD codes from elliptic curves, hyperelliptic curves, and Hermitian curves.
Later, Bhowmick, Dalai, and Mesnager \cite{bhowmickLinearComplementary2024} established the characterization and construction mechanism of LCPs of AG codes, and proposed explicit LCPs of codes from elliptic curves.
These construction schemes for LCPs of AG codes and LCD AG codes are based on non-special divisors of degrees $g$ and $g-1$, where $g$ is the genus of a function field. 
In 2024, Moreno, L{\'o}pez, and Matthews \cite{morenoExplicitNonspecial2024} explicitly presented non-special divisors of degrees $g$ and $g-1$ on Kummer extensions $\mathbb{F}_q(x,y)/\mathbb{F}_q(x)$ defined by 
$$y^m = \prod_{i=1}^{r}(x-\alpha_i),$$
where $\gcd(m,r) = 1$ and $\alpha_1,\dots,\alpha_r\in\mathbb{F}_q$ are pairwise distinct. The support of these divisors consists of the zeros of $y$. Moreover, they provided LCD AG codes on the Hermitian function field.
In \cite{castellanosLinearComplementary2025}, Castellanos, Marques, and Quoos proposed more constructions of LCPs of AG codes and LCD AG codes on the same Kummer extensions with $r<m$, hyperelliptic curves, and elliptic curves.

In the past year, numerous studies have emerged on non-special divisors of small degree and the construction of LCPs of AG codes on the Kummer extension $\mathbb{F}_q(x,y)/\mathbb{F}_q(x)$ defined by 
\begin{equation}\label{equation:Kummer_positive}
y^m = \prod_{i=1}^{r}(x-\alpha_i)^{\lambda_i},
\end{equation}
where $\alpha_1,\dots,\alpha_r\in\mathbb{F}_q$ are pairwise distinct and $\lambda_i\geq 1$ for $1\leq i\leq r$.
Using the same technique based on Weierstrass semigroups presented in \cite{morenoExplicitNonspecial2024}, Huang et al. \cite{junjieLinearComplementary2025} obtained non-special divisors of degrees $g$ and $g-1$, whose support consists of totally ramified places, and constructed LCPs of codes from subcovers of the BM curve, elliptic curves, hyperelliptic curves and other function fields.
In \cite{mendozaCharacterizationNonspecial2026}, Mendoza, Navarro, and Quoos provided an arithmetic criterion to identify all non-special divisors of degrees $g-1$ and $g$ whose support is contained in a subset of totally ramified places. They also explicitly determined all such divisors and provided new families of LCPs of AG codes in certain cases.
Using Galois theory, Marques, Silva, and Tafazolian established necessary and sufficient conditions for non-special divisors with no constraint on the support \cite{marquesConstructionNonspecial2026}. Furthermore, they explicitly computed non-special divisors whose support can contain non-totally ramified places, and constructed new families of LCPs of AG codes. 

In this work, we investigate a more general Kummer extension $\mathbb{F}_q(x,y)/\mathbb{F}_q(x)$ defined by
$$y^m = f(x),$$
where $f(x) \in \mathbb{F}_q(x)\setminus\mathbb{F}_q$. Inspired by \cite{mendozaCharacterizationNonspecial2026}, we present a generalized arithmetic characterization of non-special divisors on the Kummer extension (see Theorem \ref{theorem:non-special_thm}). 
In this characterization, a divisor can be of arbitrary degree, and its support may contain non-totally ramified places and non-rational places.
Our main idea is to unify all extensions of each place of $\mathbb{F}_q(x)$ as a whole. This unified object is considered an invariant divisor under Galois group actions, and we employ the dimension formula derived from the result in \cite{maharajCodeConstruction2004}.
This characterization extends previous findings presented in \cite{mendozaCharacterizationNonspecial2026}. 
Moreover, based on this result, we explicitly construct non-special divisors of degree $g-1$ on the function field of the GK curve (see Proposition \ref{proposition:GK_non-special}).

On the other hand, for an arbitrary function field, we show that given a canonical divisor and a non-special divisor of degree $g-1$, another non-special divisor of degree $g-1$ can be determined (see Theorem \ref{theorem:canonical_non-special}). We also present a method to characterize effective non-special divisors of degree $g$ using pure gaps (see Theorem \ref{theorem:puregap_dim}).
Based on recent results on pure gaps, we explicitly construct three families of effective non-special divisors of degree $g$ on the Kummer extension with the same multiplicities (see Propositions \ref{proposition:explicit_non-special1}, \ref{proposition:explicit_non-special2}, and \ref{proposition:explicit_non-special3}), and prove that they are inequivalent.

Finally, we propose a general framework for constructing LCPs of AG codes on Kummer extensions (see Theorem \ref{theorem:LCP_construction}). Within this framework, we utilize the fact that a non-special divisor can be determined by a canonical divisor, and subsequently show that the security parameter of the resulting LCP of codes is given by the minimum distance of one of the codes. 
Moreover, we construct LCD AG codes under specific conditions. In particular, we derive the construction of LCD codes on Kummer extensions with linearized polynomials (see Proposition \ref{proposition:LCD_linearized}).
We demonstrate our results with numerous concrete examples. Examples include the construction of LCPs of AG codes on the function field of the GK curve (see Proposition \ref{proposition:GK_LCP}) and LCD AG codes on quotients of the Hermitian curve (see Corollary \ref{corollary:LCD_Hermitian}).

The paper is organized as follows. In Section \ref{section2}, we introduce some preliminaries on function fields, LCPs of codes, LCD codes, and Kummer extensions.
In Section \ref{section3}, we present a characterization of non-special divisors on Kummer extensions, and also propose methods to identify non-special divisors via canonical divisors and pure gaps on arbitrary function fields.
In Section \ref{section4}, we explicitly construct non-special divisors, including non-special divisors of degree $g-1$ on the GK curve and three families of effective non-special divisors of degree $g$ on Kummer extensions with the same multiplicities.
In Section \ref{section5}, we provide a framework for the construction of LCPs of codes and LCD codes on Kummer extensions, and present a large number of examples.

\section{Preliminaries}\label{section2}
In this section, we briefly introduce notations and results related to function fields, algebraic geometry codes, LCPs of codes, LCD codes, and Kummer extensions.
Throughout this article, let $q$ be a power of a prime number $p$, and let $\mathbb{F}_q$ be the finite field with $q$ elements. Let $\mathbb{N} := \{0,1,2,3,\cdots\}$.
For $a,b\in\mathbb{Z}$, we denote by $\gcd(a,b)$ the greatest common divisor of $a$ and $b$.
For $c\in \mathbb{R}$, we denote by $\floor{c}$ the largest integer not greater than $c$ and by $\ceil{c}$ the smallest integer not less than $c$. 
The following useful lemma regarding the floor function $\floor*{\cdot}$.
\begin{lemma}\cite[Page 94]{grahamConcreteMathematics1994}\label{lemma:floor_sum}
	Let $a,c\in\mathbb{Z}$ and $b\in\mathbb{N}$ satisfy $b\geq 2$. Then 
	$$\sum_{i=0}^{b-1}\floor*{\frac{ia + c}{b}} = \gcd(a,b)\floor*{\frac{c}{\gcd(a,b)}} + \frac{ab - a - b + \gcd(a,b)}{2}.$$
	In particular, we have that:

	(i) If $c=0$, then 
	$$\sum_{i=0}^{b-1}\floor*{\frac{ia}{b}} = \sum_{i=1}^{b-1}\floor*{\frac{ia}{b}} = \frac{ab - a - b + \gcd(a,b)}{2}.$$
	
	(ii) If $a = 1$, then 
	$$\sum_{i=0}^{b-1}\floor*{\frac{c + i}{b}} = c.$$
\end{lemma}

\subsection{Function Fields}
Let $F/\mathbb{F}_q$ be a function field of genus $g$ with full constant field $\mathbb{F}_q$. We denote by $\mathbb{P}_F$ the set of places of $F$, by $\Omega_F$ the module of Weil differentials of $F$, by $v_P$ the discrete valuation of $F$ with respect to the place $P\in \mathbb{P}_F$, and by $\operatorname{Div} (F)$ the free abelian group generated by the places of $F$.
A place of degree one is called a rational place of $F$. The set of all rational places of $F$ is denoted by $\mathbb{P}_{F}^1$. An element $D\in \operatorname{Div}(F)$ is called a divisor of $F$ and its degree is given by $\deg D := \sum_{P\in\operatorname{supp} D} v_P(D)\cdot \deg P$, where $\operatorname{supp} D$ is the support of $D$.
A divisor $D$ is said to be effective if $v_P(D)\geq 0$ for all $P\in \mathbb{P}_F$.
For a non-zero element $z\in F$, we denote by $(z)$, $(z)_\infty$, and $(z)_0$ the principal divisor, the pole divisor and the zero divisor of $z$, respectively.

For a non-zero element $\omega\in\Omega_F$, we denote by $(\omega)$ the canonical divisor corresponding to $\omega$.
Two divisors $D_1, D_2\in \operatorname{Div}(F)$ are said to be equivalent, denoted by $D_1\sim D_2$, if there is a function $z\in F$ such that $D_1 - D_2 = (z)$. 
All canonical divisors of $F$ are equivalent since $\Omega_F$ is a one-dimensional vector space over $F$. For two equivalent divisors $D_1$ and $D_2$, we have $\ell(D_1) = \ell(D_2)$ and $\deg D_1 = \deg D_2$.

Given a divisor $D\in \operatorname{Div}(F)$, the Riemann-Roch space associated to $D$ is defined by 
$$\mathcal{L}(D) := \{z\in F\mid (z) \geq -D\}\cup\{0\}.$$
It is a vector space of finite dimension over $\mathbb{F}_q$. We denote by $\ell(D)$ the dimension of $\mathcal{L}(D)$ over $\mathbb{F}_q$.
Let $W$ be a canonical divisor of $F$. Then for each divisor $D\in \operatorname{Div}(F)$, the Riemann-Roch Theorem says that 
$$\ell(D) = \deg D + 1 - g + \ell(W-D).$$
Since $\ell(W-D) \geq 0$, we have $\ell(D)\geq \deg D + 1 - g$. The divisor $D$ is called non-special if 
$$\ell(D) = \deg D + 1 - g .$$
Otherwise, the divisor $D$ is called special. A non-special divisor of degree $g-1$ can be derived from an effective non-special divisor of degree $g$.

\begin{lemma}\cite[Lemma 3]{balletExistenceNonspecial2006}\label{lemma:non-special_g_to_g-1}
	Assume that $D\in \operatorname{Div}(F)$ is an effective non-special divisor of degree $g$. If there exists a rational place $P\in \mathbb{P}_F^1$ such that $P\notin\operatorname{supp}D$, then $D-P$ is a non-special divisor of degree $g-1$.
\end{lemma}

Next, we introduce the theory of Weierstrass semigroups. Let $P_1,\dots,P_s$ be $s$ distinct rational places of $F$. The Weierstrass semigroup at $P_1,\dots,P_s$ is defined by 
$$H(P_1,\cdots,P_s) := \left\{(a_1,\dots,a_s)\in \mathbb{N}^s\mid\exists~z\in F \text{~with~}(z)_0 = \sum_{i=1}^s a_iP_i\right\}.$$
The complementary set $G(P_1,\dots,P_s):= \mathbb{N}^s \setminus H(P_1,\dots, P_s)$ is called the set of gaps at $P_1,\dots,P_s$. 
There exists a special class of gaps known as pure gaps. One can use dimensions of Riemann-Roch spaces to characterize pure gaps.
An $s$-tuple $(a_1,\dots,a_s)\in\mathbb{N}^s$ is a pure gap at $P_1,\dots P_s$ if and only if
$$\ell\left(\sum_{i=1}^s a_i P_i\right) = \ell\left(\sum_{i=1}^{s} a_iP_i-P_j\right) \text{ for all } 1\leq j\leq s.$$
The set of all pure gaps at $P_1,\dots,P_s$ is denoted by $G_0(P_1,\dots,P_s)$. In particular, when $s=1$, we have $G(P_1) = G_0(P_1)$.
In fact, the set $G(P_1)$ contains exactly $g$ elements satisfying $1 = i_1<i_2<\dots <i_g\leq 2g-1$.
The following result concerns pure gaps at several places.
\begin{proposition}\cite[Proposition 2.1]{zhangPureGaps2026}\label{proposition:subpuregap} 
	Suppose that $s\geq 2$, $1\leq t\leq s$, and $\{i_1,\dots,i_t\}\subseteq \{1,\dots,s\}$.
	Let $(a_1,\dots,a_s)\in G_0(P_1,\dots,P_s)$.
	Then $(a_{i_1},\dots,a_{i_t})\in G_0(P_{i_1},\dots,P_{i_t})$.
\end{proposition}
Further investigations on the theory of Weierstrass semigroups on function fields can be found in \cite{kimIndexWeierstrass1994,matthewsWeierstrassSemigroup2004,carvalhoGoppaCodes2005,castellanosOneTwoPoint2016,castellanosCompleteSet2024,castellanosGeneralizedWeierstrass2026} and the references cited therein.

\subsection{LCPs of AG Codes and LCD AG codes}
A linear code $\mathcal{C}$ over $\mathbb{F}_q$ is a linear subspace of $\mathbb{F}_q^n$. We call $n$ the length of $\mathcal{C}$ and $k=\dim_{\mathbb{F}_q}\mathcal{C}$ the dimension of $\mathcal{C}$.
The minimum distance of $\mathcal{C}$ is defined as 
$$d(\mathcal{C}):=\min \{\operatorname{wt}(\mathbf{c})\mid 0\neq \mathbf{c}\in \mathcal{C}\},$$
where $\operatorname{wt}(\mathbf{c}):=\#\{i\mid c_i\neq 0\}$ for $\mathbf{c} = (c_1,\dots,c_n)\in\mathbb{F}_q^n$.
We call $[n,k,d(\mathcal{C})]$ the parameters of a linear code $\mathcal{C}$. Let $\langle\cdot,\cdot\rangle$ be the standard inner product on $\mathbb{F}_q^n$. Then the code 
$$\mathcal{C}^\bot := \{\mathbf{u}\in\mathbb{F}_q^n\mid \langle \mathbf{u}, \mathbf{c}\rangle = 0 \text { for all } \mathbf{c}\in\mathcal{C}\}$$
is called the dual code of $\mathcal{C}$. 
\begin{definition}
    Two linear codes $\mathcal{C}_1, \mathcal{C}_2\subseteq\mathbb{F}_q^n$ are said to be equivalent, denoted by $\mathcal{C}_1\sim \mathcal{C}_2$, if there exists a vector $\mathbf{a} = (a_1,\dots,a_n)\in(\mathbb{F}_q^\ast)^n$ such that 
    $$\mathcal{C}_2  = \mathbf{a}\cdot \mathcal{C}_1:=\{(a_1c_1,\dots,a_nc_n)\mid (c_1,\dots,c_n)\in\mathcal{C}_1\}.$$
\end{definition}
Note that two equivalent codes have the same length, dimension and minimum distance.
\begin{definition}
    A pair of linear codes $(\mathcal{C}_1,\mathcal{C}_2)$ of length $n$ over $\mathbb{F}_q$ is called a linear complementary pair (LCP) if
    $$\mathcal{C}_1\oplus \mathcal{C}_2 = \mathbb{F}_q^n.$$
\end{definition}
Given an LCP of codes $(\mathcal{C}_1,\mathcal{C}_2)$, the minimum distance $d(\mathcal{C}_1)$ measures resistance against FIA, and $d(\mathcal{C}_2^\bot)$ measures resistance against SCA.
The joint security against the two attacks is provided by $\min\{d(\mathcal{C}_1),d(\mathcal{C}_2^\bot)\}$, which is called the security parameter of $(\mathcal{C}_1,\mathcal{C}_2)$.
\begin{definition}
    A linear code $\mathcal{C}$ of length $n$ over $\mathbb{F}_q$ is called a linear complementary dual (LCD) code if 
    $$\mathcal{C}\oplus\mathcal{C}^\bot = \mathbb{F}_q^n.$$
\end{definition}
Observe that if $\mathcal{C}$ is an LCD code, then $(\mathcal{C},\mathcal{C}^\bot)$ is an LCP of codes with the security parameter $d(\mathcal{C})$.

Next we introduce algebraic geometry (AG) codes. Let $F/\mathbb{F}_q$ be a function field of genus $g$ with full constant field $\mathbb{F}_q$. Let $P_1,\dots,P_n$ be $n$ pairwise distinct rational places of $F$ and $D = \sum_{i=1}^n P_i$.
Suppose that $G\in \operatorname{Div}(F)$ is a divisor satisfying $\operatorname{supp} G\cap\operatorname{supp}D = \varnothing$. The AG code $C_\mathcal{L}(D,G)$ associated with $D$ and $G$ is defined as 
$$C_\mathcal{L}(D,G):=\{(f(P_1),\dots,f(P_n))\mid f\in\mathcal{L}(G)\}.$$
The following theorem determines the parameters of the AG code $C_\mathcal{L}(D,G)$.
\begin{theorem}\cite[Corollary 2.2.3]{stichtenothAlgebraicFunctionFields2009}\label{theorem:AG_codes}
	Suppose that $2g-2<\deg G <n$. Then the AG code $C_\mathcal{L}(D,G)$ has parameters $[n, \deg G+1-g, \geq n-\deg G]$.
\end{theorem}
The dual code of the AG code $C_\mathcal{L}(D,G)$ is still an AG code. Let $\omega$ be a Weil differential such that $v_{P_i}(\omega) = -1$ for all $i=1,\dots,n$. Then 
	$$C_\mathcal{L}(D,G)^\bot = \mathbf{a}\cdot C_\mathcal{L}(D,H),$$
	where $H = D-G+(\omega)$ and $\mathbf{a} = (\operatorname{res}_{P_1}(\omega),\dots,\operatorname{res}_{P_n}(\omega))$.

For two divisors $A,B\in \operatorname{Div}(F)$, define the greatest common divisor of $A$ and $B$ by 
$$\gcd(A,B) := \sum_{P\in\mathbb{P}_F} \min\{v_P(A),v_P(B)\}P,$$
and the least multiple divisor of $A$ and $B$ by 
$$\lmd(A,B) := \sum_{P\in\mathbb{P}_F} \max\{v_P(A),v_P(B)\}P.$$

In \cite{bhowmickLinearComplementary2024}, Bhowmick, Dalai and Mesnager presented some conditions for constructing LCPs of AG codes. 
\begin{theorem}\cite[Theorem 3.5]{bhowmickLinearComplementary2024}\label{theorem:LCP}
	Let $C_\mathcal{L}(D,G)$ and $C_\mathcal{L}(D,H)$ be two AG codes over the function field $F/\mathbb{F}_q$ of genus $g\neq 0$, where $\operatorname{supp}G\cap\operatorname{supp} D =\operatorname{supp}H\cap\operatorname{supp} D =  \varnothing$. Suppose that the divisors $G$ and $H$ satisfy 

	(i) $2g-2<\deg G, \deg H <n$, $\ell(G) + \ell(H) = n$, 

	(ii) $\deg (\gcd(G,H)) = g-1$, and 
	
	(iii) both divisors $\gcd(G,H)$ and $\lmd (G,H) - D$ are non-special.

	Then the pair $(C_\mathcal{L}(D,G),C_\mathcal{L}(D,H))$ is an LCP of AG codes.
\end{theorem}

In \cite{mesnagerComplementaryDual2018}, Mesnager, Tang and Qi provided a construction scheme for obtaining algebraic geometry LCD codes.
\begin{theorem}\cite[Theorem 6]{mesnagerComplementaryDual2018}\label{theorem:LCD}
	Suppose that $C_\mathcal{L}(D,G)$ is an AG code over the function field $F/\mathbb{F}_q$ of genus $g$.
	Let $\omega$ be a Weil differential such that $(\omega) = G+H-D$ for some $H$ with $\operatorname{supp}G\cap\operatorname{supp} D =\operatorname{supp}H\cap\operatorname{supp} D =  \varnothing$.
	Assume that $\operatorname{res}_{P_i}(\omega) = \operatorname{res}_{P_j}(\omega)$ for all $1\leq i<j\leq n$ and $\gcd(G,H)$ is a divisor of degree $g-1$. 
	Then the AG code $C_\mathcal{L}(D,G)$ is an LCD code if and only if $\gcd(G,H)$ is non-special.
\end{theorem}

\subsection{Kummer Extensions}
Let $m\geq 2$ be an integer such that $\gcd(m,q) = 1$. Let $f(x)\in \mathbb{F}_q(x)\setminus \mathbb{F}_q$ satisfy that $f(x)$ is not a $d$-th power of an element in $\mathbb{F}_q(x)$ for all $d\mid m$, $d>1$.
Consider the Kummer extension $F = \mathbb{F}_q(x,y)/\mathbb{F}_q(x)$ defined by the affine equation
\begin{equation}\label{equation:Kummer_extension}
	y^m = f(x).
\end{equation}
Note that $\varphi(T) = T^m - f(x)\in \mathbb{F}_q(x)[T]$, which is the minimal polynomial of $y$ over $\mathbb{F}_q(x)$, is also irreducible in $\overline{\mathbb{F}}_q(x)[T]$. It follows that $\mathbb{F}_q$ is the full constant field of $F$.
Let $\operatorname{Con}_{F/\mathbb{F}_q(x)}$ be the conorm with respect to $F/\mathbb{F}_q(x)$.
We denote by $P_0,P_1,\dots,P_r\in \mathbb{P}_{\mathbb{F}_q(x)}$ the places associated to the zeros and poles of $f(x)$.
For each $i\geq r+1$, we denote by $P_i\in \mathbb{P}_{\mathbb{F}_q(x)}$ a place distinct from $P_0,P_1,\dots,P_r$.
For each $i\in\mathbb{N}$, let $\lambda_i:= v_{P_i}(f(x))$, $d_i:=\deg P_i$, $r_i := \gcd(m,\lambda_i)$, and 
$$D_i:= \frac{r_i}{m}\operatorname{Con}_{F/\mathbb{F}_q(x)}(P_i) = \sum_{Q\in\mathbb{P}_F, Q\mid P_i}Q.$$
We have that $\lambda_i = 0$ for $i\geq r+1$ and $\sum_{i=0}^{r}\lambda_id_i = 0$.
For any integer $n\geq r$, by \cite[Proposition 3.7.3]{stichtenothAlgebraicFunctionFields2009}, the genus of $F$ is given by
    $$g = \frac{2-2m+\sum_{i=0}^{n} \left(m-\gcd(m,\lambda_i)\right) \cdot d_i}{2},$$
and we have the following principal divisor
\begin{equation}\label{equation:y_divisor}
	(y)=\sum_{i=0}^{n}\frac{\lambda_i}{r_i}D_i.
\end{equation}
Moreover, we have a canonical divisor
\begin{equation}\label{equation:dx_divisor}
    (dx) = -2(x)_\infty + \operatorname{Diff}(F/\mathbb{F}_q(x)) =  -2(x)_\infty + \sum_{i=0}^{r}\left(\frac{m}{r_i}-1\right)D_i.
\end{equation}
\begin{lemma}\cite[Lemma 3.1]{zhangWeierstrassSemigroups2026}\label{Lemma:sum_mod}
	Suppose that $(a_0,a_1,\dots,a_r)\in \mathbb{Z}^{r+1}$, $b, c\in\mathbb{Z}$ and $b \equiv c\pmod m$. Then
	\begin{equation*}
	\sum_{i=0}^r\left\lfloor\frac{a_i+b\lambda_i}{m}\right\rfloor d_i = \sum_{i=0}^r\floor*{\frac{a_i+c\lambda_i}{m}} d_i.
	\end{equation*}
\end{lemma}
\begin{remark}
    In the classical framework, the extension $F/\mathbb{F}_q(x)$ is said to be a Kummer extension under the assumption that $\mathbb{F}_q$ contains a primitive $m$-th root of unity. In this paper, we remove this assumption and still call $F/\mathbb{F}_q(x)$ a Kummer extension.
\end{remark}
If $\mathbb{F}_q$ contains a primitive $m$-th root of unity, then $F/\mathbb{F}_q(x)$ is a Galois extension. The following lemma is obtained from \cite[Theorem 2.2]{maharajCodeConstruction2004}.
\begin{lemma}\cite[Lemma 3.3]{zhangWeierstrassSemigroups2026}\label{lemma:D_sumdim}
    Suppose that $K$ is an algebraic extension of $\mathbb{F}_q$ such that $K$ contains a primitive $m$-th root of unity. Let $F' = FK$ be the constant field extension of $F$, and let $\operatorname{Con}_{F'/F}(\cdot)$ be the conorm with respect to $F'/F$.
    Let $D$ be a divisor of $F$ such that $\operatorname{Con}_{F'/F}(D)$ is invariant under the action of the Galois group $\operatorname{Gal}(F'/K(x))$. Then
    \begin{equation*}\label{equation:D_dim}
        \ell(D) = \sum_{t=0}^{m-1}\ell\left(\left[D+(y^t)\right]\Big|_{\mathbb{F}_q(x)}\right).
    \end{equation*}
\end{lemma}

In particular, suppose that the Kummer extension $F = \mathbb{F}_q(x,y)/\mathbb{F}_q(x)$ is defined by 
\begin{equation}\label{equation:Kummer_same_multiplicities}
	f(x) = \prod_{i=1}^{r}(x-\alpha_i)^\lambda,
\end{equation}
where $\gcd(m,r\lambda)=1$ and $\alpha_1,\dots,\alpha_r$ are pairwise distinct elements of $\mathbb{F}_q$. Let $P_i\in \mathbb{P}_{\mathbb{F}_q(x)}$ be the zero of $x-\alpha_i$ and let $P_0\in \mathbb{P}_{\mathbb{F}_q(x)}$ be the pole of $x$.
Then $P_i$ is totally ramified in $F/\mathbb{F}_q(x)$ and there is exactly one place $Q_i\in\mathbb{P}_F$ lying over $P_i$ for $0\leq i\leq r$. We have that $D_i = Q_i$ for $1\leq i\leq r$ and $Q_i$ is rational for $0\leq i\leq r$. 
There are some results on the Weierstrass semigroup at $Q_0$, gaps at $Q_i$ for $1\leq i\leq r$, and pure gaps at $Q_0,Q_1,\dots,Q_r$.
\begin{proposition}\cite[Theorem 3.2]{castellanosOneTwoPoint2016}\label{proposition:Weierstrass_Q_0}
	The Weierstrass semigroup at $Q_0$ is $\langle m,r\rangle$.
\end{proposition}
\begin{proposition}{\cite[Corollary 2]{huMultipointCodesKummer2018}}\label{proposition:special_gap}
	For every $1\leq i\leq r$,
	$$G(Q_i) = \left\{ mk+j ~| ~0\leq k\leq r-2-\floor*{\frac{r}{m}},~1\leq j\leq m-\ceil*{\frac{(k+1)m}{r}} \right\}.$$
\end{proposition}

\begin{proposition}\cite[Proposition 4.6]{zhangPureGaps2026}\label{proposition:puregap_1}
	Suppose that $2\leq s\leq r-\floor*{\frac{r}{m}}-1$ and $t_{i} = \floor*{\frac{m(r-i)}{r}}$ for $1\leq i\leq s$.
	Then
	$$\{(j_1,\dots,j_s) ~|~ 1\leq j_{i}\leq t_{i} \text{ for } 1\leq i\leq s\}\subseteq G_0(Q_1,\dots,Q_s).$$
\end{proposition}

\begin{proposition}\cite[Proposition 4.10]{zhangPureGaps2026}\label{proposition:puregap_2}
	Suppose that $1\leq s\leq r$. 
	Let $1\leq k \leq r-1$ and $1\leq j_0\leq m-1$ such that $mk-rj_0\in G(Q_0)$.
	For $1\leq i\leq s$, let 
	$$t_i = \left\{\begin{array}{cc}
		m-\ceil*{\frac{(k+i)m}{r}}+j_0,~&\text{if}~ k+i \neq r,\\
		m-\ceil*{\frac{(k+i)m}{r}}+j_0-1,~&\text{if}~ k+i = r.\\
	\end{array}\right.$$
	Then 
	$$\{(mk-rj_0,j_1,\dots,j_s)~|~ 1\leq j_{i}\leq t_i \text{ for } 1\leq i\leq s\}\subseteq G_0(Q_0,Q_1,\dots,Q_s).$$
\end{proposition}

\section{Characterization of Non-special Divisors}\label{section3}
In this section, we consider how to determine non-special divisors on a function field. 
Firstly, on a Kummer extension $F = \mathbb{F}_q(x,y)/\mathbb{F}_q(x)$ defined by \eqref{equation:Kummer_extension}, we give an arithmetic criterion to determine all non-special divisors of the form $\sum_{i=0}^{n}a_i D_i$ for any integer $n\geq r$.
Then, on a function field of genus $g$, we point out that another non-special divisor of degree $g-1$ can be derived from a known non-special divisor of degree $g-1$ and a known canonical divisor. Moreover, we show that pure gaps can be used to identify effective non-special divisors of degree $g$.

\subsection{Characterization of Non-special Divisors on Kummer Extensions}
In this subsection, for the Kummer extension $F = \mathbb{F}_q(x,y)/\mathbb{F}_q(x)$ defined by \eqref{equation:Kummer_extension}, we present our main result as follows, which characterizes all non-special divisors generated by $D_1,\dots,D_n$.
\begin{theorem}\label{theorem:non-special_thm}
	For any integer $n \geq r$, let $(a_0,a_1,\dots,a_n)\in\mathbb{Z}^{n+1}$. Then $D = \sum_{i=0}^{n}a_iD_i$ is a non-special divisor if and only if 
	\begin{equation}\label{condition}
		\ell(D) = \sum_{t=0}^{m-1}\left(\sum_{i=0}^{n}\floor*{\frac{r_ia_i+t\lambda_{i}}{m}} d_{i} +1\right).
	\end{equation}
\end{theorem}

In order to prove the theorem above, we provide a lemma, which presents an equation related to the genus $g$ of $F$.
\begin{lemma}\label{lemma:floor_sum_g}
	For any integer $n\geq r$,
	$$\sum_{i=0}^{n}\sum_{t=0}^{m-1}\floor*{\frac{t\lambda_i}{m}}d_i+m-1 = -g.$$
\end{lemma}
\begin{proof}
	We have that
	\begin{align*}
		&~\sum_{i=0}^{n}\sum_{t=0}^{m-1}\floor*{\frac{t\lambda_i}{m}}d_i + m-1\\
		=&~\sum_{i=0}^{n}\frac{m\lambda_i-m-\lambda_i+\gcd(m,\lambda_i)}{2}\cdot d_i +m-1 ~~~ (\text{from Lemma \ref{lemma:floor_sum} (i)})\\
		=&~\frac{-m\sum_{i=0}^{n}d_i + \sum_{i=0}^{n}\gcd(m,\lambda_i)d_i}{2} + m-1 ~~~\left(\text{from $\sum_{i=0}^{n}\lambda_i d_i = 0$}\right)\\
		=&~-\frac{2-2m+\sum_{i=0}^{n} \left(m-\gcd(m,\lambda_i)\right) \cdot d_i}{2} = -g
	\end{align*}
\end{proof}

Now we proceed to prove Theorem \ref{theorem:non-special_thm}.
\begin{proof}[Proof of Theorem \ref{theorem:non-special_thm}]
	Note that $\deg D = \sum_{i=0}^{n}r_ia_id_i$. We will show that 
	\begin{equation}\label{equation:iff}
		\sum_{t=0}^{m-1}\left(\sum_{i=0}^{n}\floor*{\frac{r_ia_i+t\lambda_{i}}{m}} d_{i} +1\right) = \deg D - g+1.
	\end{equation}
	For each $0\leq t\leq m-1$ and $0\leq i\leq n$, let $0\leq t_i\leq m-1$ such that $t\lambda_i = m\floor*{\frac{t\lambda_i}{m}} + t_i$.
	Then we obtain that 
	$$ t\lambda_i/r_i  =  \frac{m}{r_i}\floor*{\frac{t\lambda_i/r_i}{m/r_i}} + t/r_i.$$
	Thus $t/r_i$ is an integer with $0\leq t/r_i\leq m/r_i-1$.
	Since $\gcd(m/r_i, \lambda_i/r_i) = 1$, we have that $t_i/r_i$ runs $r_i$ times over $\{0,1,\dots, m/r_i-1\}$ as $t$ ranges from $0$ to $m-1$.
	Thus, we get that
	\begin{align*}
		  &~\sum_{t=0}^{m-1}\left(\sum_{i=0}^{n}\floor*{\frac{r_ia_i+t\lambda_{i}}{m}} d_{i} +1\right)\\
		= &~\sum_{t=0}^{m-1}\left(\sum_{i=0}^{n}\left(\floor*{\frac{r_ia_i+t_i}{m}} + \floor*{\frac{t\lambda_i}{m}}\right) d_{i} +1\right)\\
		= &~\sum_{i=0}^{n}\sum_{t=0}^{m-1}\floor*{\frac{a_i + t_i/r_i}{m/r_i}}d_i + \sum_{i=0}^{n}\sum_{t=0}^{m-1}\floor*{\frac{t\lambda_i}{m}}d_i+m\\
		= &~\sum_{i=0}^{n}r_id_i\sum_{t=0}^{m/r_i-1}\floor*{\frac{a_i + t}{m/r_i}} + \sum_{i=0}^{n}\sum_{t=0}^{m-1}\floor*{\frac{t\lambda_i}{m}}d_i + m\\
		= &~\sum_{i=0}^{n}r_ia_id_i + \sum_{i=0}^{n}\sum_{t=0}^{m-1}\floor*{\frac{t\lambda_i}{m}}d_i + m -1 + 1 ~~~ (\text{from Lemma \ref{lemma:floor_sum} (ii)})\\
		= &~ \deg D-g+1  ~~~ (\text{from Lemma \ref{lemma:floor_sum_g}}).
	\end{align*}
	The conclusion follows directly from the definition of non-special divisors, which completes the proof.
\end{proof}

Similar results appear in \cite{mendozaCharacterizationNonspecial2026}, where Mendoza, Navarro, and Quoos completely characterized all non-special divisors of degrees $g-1$ and $g$ on the Kummer extension given by \eqref{equation:Kummer_positive}.
The support of these divisors consists of totally ramified places of degree one.
Here, non-special divisors characterized in Theorem \ref{theorem:non-special_thm} can be of arbitrary degree. Moreover, their support may contain non-totally ramified places as well as non-rational places.
Next we show how Theorem \ref{theorem:non-special_thm} generalizes the results presented in \cite{mendozaCharacterizationNonspecial2026}.
We first need a lemma giving a formula to compute $\ell(\sum_{i=0}^{n} a_iD_i)$.
\begin{lemma}\label{D_dim}
	For any integer $n\geq r$, let $(a_0,a_1,\dots,a_n)\in\mathbb{Z}^{n+1}$. Then 
	$$\ell\left(\sum_{i=0}^{n} a_i D_{i}\right) = \sum_{t=0}^{m-1}\max\left\{0,~\sum_{i=0}^{n}\floor*{\frac{r_ia_i+t\lambda_{i}}{m}} d_{i} +1\right\}.$$
\end{lemma}
\begin{proof}
	For each $t\in\{0,1,\dots,m-1\}$, by the divisor \eqref{equation:y_divisor} we have 
	$$\sum_{i=0}^{n} a_i D_i + (y^t) = \sum_{i=0}^{n}\left(a_i + \frac{t\lambda_i}{r_i}\right)D_i.$$
	This implies that 
	$$\left(\sum_{i=0}^{n} a_i D_i + (y^t)\right)\Big|_{\mathbb{F}_q(x)} = \sum_{i=0}^{n}\floor*{\frac{r_ia_i + t\lambda_{i}}{m}}P_{i}.$$
	By Lemma \ref{lemma:D_sumdim}, we obtain that 
	$$\ell\left(\sum_{i=0}^{n} a_i D_i\right) = \sum_{t=0}^{m-1}\ell\left(\sum_{i=0}^{n}\floor*{\frac{r_ia_i + t\lambda_{i}}{m}}P_{i} \right).$$
	It follows from the Riemann-Roch Theorem that 
	$$\ell\left(\sum_{i=0}^{n} a_i D_i\right) = \sum_{t=0}^{m-1}\max\left\{0,~\sum_{i=0}^{n}\floor*{\frac{r_ia_i+t\lambda_{i}}{m}} d_{i} + 1\right\}.$$
\end{proof}
From the above lemma, if $D = \sum_{i=0}^{n} a_iD_i$ is a non-special divisor as in Theorem \ref{theorem:non-special_thm}, then
\begin{equation}\label{equation:big_then_zero}
	\sum_{i=0}^{n}\floor*{\frac{r_ia_i+t\lambda_{i}}{m}} d_{i} +1\geq 0
\end{equation}
for all $0\leq t\leq m-1$. Based on this observation, we can characterize all non-special divisors of degrees $g-1$ and $g$, which are generated by $D_0, D_1,\dots,D_n$.
The following corollary generalizes the result presented in \cite[Theorem 3.2]{mendozaCharacterizationNonspecial2026}.

\begin{corollary}\label{corollary:non_special_g-1}
	For any integer $n \geq r$, let $(a_0,a_1,\dots,a_n)\in\mathbb{Z}^{n+1}$. Then $D = \sum_{i=0}^{n}a_iD_i$ is a non-special divisor of degree $g-1$ if and only if 
	\begin{align}
		\sum_{i=0}^{n}\floor*{\frac{r_ia_i+t\lambda_{i}}{m}} d_{i} + 1 = 0 \text{ for all } 0\leq t\leq m-1. \label{condition1}	
	\end{align}
\end{corollary}
\begin{proof}
	Assume that $D$ is a non-special divisor of degree $g-1$. Then $\ell(D) = \deg D - g+1 =0$.
	By \eqref{equation:big_then_zero} and Theorem \ref{theorem:non-special_thm}, we obtain that 
	\begin{align*}
		\sum_{i=0}^{n}\floor*{\frac{r_ia_i+t\lambda_{i}}{m}} d_{i} + 1 = 0 \text{ for all } 0\leq t\leq m-1.
	\end{align*}

	Conversely, assume that the condition \eqref{condition1} is satisfied. Then by Lemma \ref{D_dim}, we have
	\begin{align*}
		\ell (D) = \sum_{t=0}^{m-1}\left(\sum_{i=0}^{n}\floor*{\frac{r_ia_i+t\lambda_{i}}{m}} d_{i} +1\right) = 0.
	\end{align*}
	It follows from Theorem \ref{theorem:non-special_thm} that $D$ is a non-special divisor and $\deg D = \ell(D) + g-1 = g-1$.
\end{proof}

The next corollary generalizes the result presented in \cite[Theorem 3.6]{mendozaCharacterizationNonspecial2026}, which characterizes all non-special divisors $\sum_{i=1}^{n}a_iD_i$ of degree $g$.
\begin{corollary}
	For any integer $n \geq r$, let $(a_0,a_1,\dots,a_n)\in\mathbb{Z}^{n+1}$. Then $D = \sum_{i=0}^{n}a_iD_i$ is a non-special divisor of degree $g$ if and only if 
	\begin{align}
		&\sum_{i=0}^{n}\floor*{\frac{r_ia_i+t'\lambda_{i}}{m}} d_{i} = 0 \text{ for some } 0\leq t'\leq m-1, \text{ and }  \label{condition2}\\
		&\sum_{i=0}^{n}\floor*{\frac{r_ia_i+t\lambda_{i}}{m}} d_{i} + 1 = 0 \text{ for all } 0\leq t\leq m-1 \text{ with } t\neq t'. \label{condition3}	
	\end{align}
\end{corollary}
\begin{proof}
	Assume that $D$ is a non-special divisor of degree $g$. Then $\ell(D) = \deg D - g +1 = 1$.
	By \eqref{equation:big_then_zero} and Theorem \ref{theorem:non-special_thm}, we obtain that  
	\begin{align*}
		&\sum_{i=0}^{n}\floor*{\frac{r_ia_i+t'\lambda_{i}}{m}} d_{i} + 1 = 1 \text{ for some } 0\leq t'\leq m-1, \text{ and }\\
		&\sum_{i=0}^{n}\floor*{\frac{r_ia_i+t\lambda_{i}}{m}} d_{i} + 1 = 0 \text{ for all } 0\leq t\leq m-1 \text{ with } t\neq t'.
	\end{align*}
	Thus, the conditions \eqref{condition2} and \eqref{condition3} are satisfied.

	Conversely, assume that the conditions \eqref{condition2} and \eqref{condition3} are satisfied. Then by Lemma \ref{D_dim}, we have
	\begin{align*}
		\ell(D) = \sum_{t=0}^{m-1}\left(\sum_{i=0}^{n}\floor*{\frac{r_ia_i+t\lambda_{i}}{m}} d_{i} +1\right) = 1
	\end{align*}
	It follows from Theorem \ref{theorem:non-special_thm} that $D$ is a non-special divisor and $\deg D = \ell(D) + g-1 = g$.
\end{proof}

\subsection{Determining Non-special Divisors by Canonical Divisors and Pure Gaps}
In this subsection, let $F$ be a function field of genus $g$. First, we show that a non-special divisor of degree $g-1$ and a canonical divisor can determine another non-special divisor of degree $g-1$.
\begin{theorem}\label{theorem:canonical_non-special}
	Suppose that $A\in \operatorname{Div}(F)$ is a non-special divisor of degree $g-1$. Then for each canonical divisor $W\in\operatorname{Div}(F)$, the divisor $W-A$ is also a non-special divisor of degree $g-1$. 
\end{theorem}
\begin{proof}
	By the Riemann-Roch Theorem, there is a canonical divisor $W'$ such that 
	$$\ell(A) = \deg A + 1 - g + \ell(W'-A).$$
	Since $\ell(A) = \deg A + 1 -g$, we get that $\ell(W'-A) = 0$. On the other hand, we have $\deg (W'-A) = g-1$. Then $W'-A$ is a non-special divisor of degree $g-1$. 
	Thus $W-A\sim W'-A$ is also a non-special divisor of degree $g-1$ for any canonical divisor $W$.
\end{proof}

In \cite[Proposition 5]{morenoExplicitNonspecial2024}, Moreno, L{\'o}pez, and Matthews related Weierstrass semigroups to effective non-special divisors of degree $g$.
In the next theorem, we show how pure gaps determine effective non-special divisors of degree $g$.
\begin{theorem}\label{theorem:puregap_dim}
     Suppose that $Q_1,\dots,Q_s$ are $s$ distinct rational places of $F$ with genus $g\geq 1$.
	 Let $(c_1,\dots,c_s)\in\mathbb{N}^s$. Suppose that $\sum_{i=1}^{s}c_i = g$ and 
	 $$\{(c_1,a_2,\dots,a_{s})\mid 1\leq a_i\leq c_i\ \text{ for } 2\leq i\leq s\}\subseteq G_0(Q_1,\dots,Q_s),$$
	 and $\{a_1\mid 1\leq a_1\leq c_1\}\subseteq G(Q_1)$.
	 Then $\sum_{i=1}^{s}c_iQ_i$ is an effective non-special divisor of degree $g$.
\end{theorem}
\begin{proof}
	By the definition of pure gaps, we have 
	\begin{align}
		\ell\left(\sum_{i=1}^{s} c_i Q_i\right) &= \ell\left(\sum_{i=1}^{s} c_i Q_i - Q_s\right)\nonumber\\ 
		&= \ell\left(\sum_{i=1}^{s} c_i Q_i - 2Q_s\right) = \dots = \ell\left(\sum_{i=1}^{s} c_i Q_i - c_sQ_s\right) = \ell\left(\sum_{i=1}^{s-1} c_i Q_i\right). \label{equation:argument_puregap}
	\end{align}
	By Proposition \ref{proposition:subpuregap}, we have 
	$$\{(c_1,a_2,\dots,a_{s-1})\mid 1\leq a_i\leq c_i \text{ for } 2\leq i\leq s-1\}\subseteq G_0(Q_1,\dots,Q_{s-1}).$$
	Then by the same argument as \eqref{equation:argument_puregap}, we obtain that $\ell(\sum_{i=1}^{s-1} c_iQ_i) = \ell(\sum_{i=1}^{s-2}c_iQ_i)$.
	Therefore, by induction, we conclude that
	$$\ell\left(\sum_{i=1}^{s}c_iQ_i\right) = \ell\left(\sum_{i=1}^{s-1}c_iQ_i\right)= \ell\left(\sum_{i=1}^{s-2}c_iQ_i\right) = \dots = \ell\left(c_1Q_1\right).$$
	Since $\{a_1\mid 1\leq a_1\leq c_1\}\subseteq G(Q_1)$, we get that 
	$$\ell\left(c_1Q_1\right) = \ell\left((c_1-1)Q_1\right)= \ell\left((c_1-2)Q_1\right) = \dots = \ell\left(0\right) =1.$$
	On the other hand, we have $\deg \left(\sum_{i=1}^{s}c_iQ_i\right) = \sum_{i=1}^{s}c_i = g$. Thus $\sum_{i=1}^{s}c_iQ_i$ is an effective non-special divisor of degree $g$.
\end{proof}

\section{Explicit Non-special Divisors of Small Degree}\label{section4}
In this section, we explicitly construct several families of non-special divisors on Kummer extensions. We first derive a family of non-special divisors of degree $g-1$ on the function field of the GK curve.
Then, for the Kummer extension defined by $y^m = f(x)^\lambda$, where $f(x)$ is separable and splits completely over $\mathbb{F}_q$ with $\gcd\big(m,\lambda\cdot\deg f(x)\big)=1$, we explicitly present three families of effective non-special divisors of degree $g$ whose support consits of totally ramified places.

\subsection{Function Fields of GK Curves}\label{subsection:GK}
In this subsection, we consider the GK curve, which is defined by the equations
	\begin{equation*}
		\begin{cases}
		Y^{q+1} = X^q+X,\\
		Z^{\frac{q^3+1}{q+1}} = Y\frac{X^{q^2}-X}{X^q+X}.
		\end{cases}
	\end{equation*}	
	The curve, introduced by Giulietti and Korchm{\'a}ros \cite{giuliettiNewFamily2009}, is maximal over $\mathbb{F}_{q^6}$.  The genus of the GK curve is $g = \frac{1}{2}(q^3+1)(q^2-2)+1$.
	A plane model for the GK curve can be given by 
    \begin{equation}\label{equation:GKcurve}
	y^{q^3+1} = (x^q+x)\left(\frac{x^{q^2}-x}{x^q+x}\right)^{q+1}.
	\end{equation}
    We have that $F = \mathbb{F}_{q^{6}}(x,y)/\mathbb{F}_{q^{6}}(x)$ is a Kummer extension. Write 
    $$(x^q+x)\left(\frac{x^{q^2}-x}{x^q+x}\right)^{q+1} = \left(\prod_{i=1}^q(x-\alpha_i)\right) \left(\prod_{i=1}^{q^2-q}(x-\beta_i)^{q+1}\right),$$
    where $\alpha_1,\dots,\alpha_q,\beta_1,\dots,\beta_{q^2-q}\in \mathbb{F}_{q^2}$.
	Let $P_0$ be the pole of $x$ in $\mathbb{P}_{\mathbb{F}_{q^{6}}(x)}$. For each $1\leq i\leq q$, let $P_i$ be the zero of $x-\alpha_i$ in $\mathbb{P}_{\mathbb{F}_{q^{6}}(x)}$.
    For each $0\leq i\leq q$, it is obvious that $P_i$ is totally ramified in $F/\mathbb{F}_{q^{6}}(x)$. 
    Denote by $Q_i$ the unique place lying over $P_i$ for $0 \leq i \leq q$.
	For each $1\leq i\leq q^2-q$, denote by $P'_i$ the zero of $x-\beta_i$ in $\mathbb{P}_{\mathbb{F}_{q^{6}}(x)}$. Let $D_i = \frac{1}{q^2-q+1}\operatorname{Con}_{F/\mathbb{F}_{q^6}(x)}(P'_i)$ for $1\leq i\leq q^2-q$.

	In \cite[Example 3.5]{mendozaCharacterizationNonspecial2026}, Mendoza, Navarro and Quoos proved that there are no non-special divisors of degree $g-1$ with support in $\{Q_0,Q_1,\dots,Q_q\}$. 
	Here, employing Theorem \ref{theorem:non-special_thm}, we give a class of non-special divisors generated by $Q_0,\dots,Q_q$, $D_1,\dots,D_{q^2-q-1}$.
	\begin{proposition}\label{proposition:GK_non-special}
		The divisor 
		$$D = \sum_{i=0}^{q}(i(q^2-q+1)-1)Q_i + \sum_{i=1}^{q^2-q-1}iD_i$$
		is a non-special divisor of degree $g-1$.
	\end{proposition}
	\begin{proof}
		Let  
		\begin{align*}
			\varphi(t) &:= \floor*{\frac{-1 -tq^3}{q^3+1}} + \sum_{i=1}^{q}\floor*{\frac{i(q^2-q+1)-1+t}{q^3+1}} + \sum_{i=1}^{q^2-q-1}\floor*{\frac{i(q+1)+t(q+1)}{q^3+1}} + \floor*{\frac{t(q+1)}{q^3+1}}+1\\
			& = -t + \floor*{\frac{t-1}{q^3+1}} + \sum_{i=1}^{q}\floor*{\frac{i(q^2-q+1)-1+t}{q^3+1}} + \sum_{i=1}^{q^2-q-1}\floor*{\frac{i+t}{q^2-q+1}} + \floor*{\frac{t}{q^2-q+1}}+1\\
			& = -t + \sum_{i=0}^{q}\floor*{\frac{i(q^2-q+1)-1+t}{q^3+1}} + \sum_{i=0}^{q^2-q-1}\floor*{\frac{i+t}{q^2-q+1}}+1.
		\end{align*}
		By Corollary \ref{corollary:non_special_g-1}, it suffices to show that $\varphi(t) = 0$ for each $0\leq t\leq q^3$.
		By Lemma \ref{Lemma:sum_mod}, we have that $\{\varphi(t)\mid 0\leq t\leq q^3\} = \{\varphi(-t)\mid 0\leq t\leq q^3\}$. 
		Hence we only need to show $\varphi(-t) = 0$ for each $0\leq t\leq q^3$.
		Write $t = k(q^2-q+1) + j$ with $0\leq j\leq q^2-q$ and $0\leq k\leq q$. Then
		\begin{align*}
			\varphi(-t) &= k(q^2-q+1)+j + \sum_{i=0}^{q}\floor*{\frac{i(q^2-q+1)-1-k(q^2-q+1)-j}{q^3+1}} \\
			&~~~~ + \sum_{i=0}^{q^2-q-1}\floor*{\frac{i-k(q^2-q+1)-j}{q^2-q+1}} +1\\
			& = k(q^2-q+1)+j + \sum_{i=0}^{q}\floor*{\frac{(i-k)(q^2-q+1)-1-j}{q^3+1}} \\
			&~~~~ - k(q^2-q)+ \sum_{i=0}^{q^2-q-1}\floor*{\frac{i-j}{q^2-q+1}}+1.
		\end{align*}
		One has 
		$$\sum_{i=0}^{q}\floor*{\frac{(i-k)(q^2-q+1)-1-j}{q^3+1}} = -k-1 \text{ and }\sum_{i=0}^{q^2-q-1}\floor*{\frac{i-j}{q^2-q+1}} = -j.$$
		Substitution yields $\varphi(-t) = k(q^2-q+1)+j-k-1-k(q^2-q)-j+1 = 0$ for all $0\leq t\leq q^3$.
	\end{proof}

\subsection{Kummer Extensions with the Same Multiplicities}
In this subsection, we consider the Kummer extension $F = \mathbb{F}_q(x,y)/\mathbb{F}_q(x)$ with the same multiplicities, which is given by \eqref{equation:Kummer_same_multiplicities}:
$$y^m = \prod_{i=1}^r(x-\alpha_i)^\lambda,$$
where $\gcd(m,r\lambda) = 1$ and $\alpha_1,\dots,\alpha_r$ are pairwise distinct elements of $\mathbb{F}_q$. Recall that $f(x)$ has exactly one pole $Q_0\in\mathbb{P}_F$ and $r$ zeros $Q_1,\dots,Q_r\in\mathbb{P}_F$ with identical multiplicities.
The genus of $F$ is given by 
$$g = \frac{(m-1)(r-1)}{2}.$$
sing Theorem \ref{theorem:puregap_dim} together with the pure-gap results from Propositions \ref{proposition:puregap_1} and \ref{proposition:puregap_2}, we construct three families of effective non-special divisors of degree $g$, whose support is contained in $\{Q_0,Q_1,\dots,Q_r\}$.  

Suppose that $\lambda = 1$ and $r<m$. In \cite[Theorem 8]{morenoExplicitNonspecial2024}, Moreno, L{\'o}pez and Matthews provided explicit effective non-special divisors of degree $g$. The form of these divisors is  
$$\sum_{i=1}^{r-1}\floor*{\frac{im}{r}}Q_i.$$
In the next proposition, we generalize such divisor forms to arbitrary $m$, $r$ and $\lambda$ with $\gcd(m,r\lambda)=1$.
Related similar results appear in \cite[Example 3.6]{marquesConstructionNonspecial2026}.
\begin{proposition}\label{proposition:explicit_non-special1}
	The divisor
	\begin{equation} \label{equation:non-special_divisor1}
		D = \sum_{i=1}^{r - \floor*{\frac{r}{m}} -1} \floor*{\frac{m(r-i)}{r}}Q_{i}
	\end{equation}
	is an effective non-special divisor of degree $g$.
\end{proposition}
\begin{proof}
	By Proposition \ref{proposition:special_gap}, we have
		$$\sum_{i=1}^{r - \floor*{\frac{r}{m}} -1} \floor*{\frac{m(r-i)}{r}}= \sum_{i=1}^{r - \floor*{\frac{r}{m}} -1} \left(m - \ceil*{\frac{im}{r}}\right)  = \sum_{i=0}^{r - \floor*{\frac{r}{m}} -2} \left(m - \ceil*{\frac{(i+1)m}{r}}\right) = g,$$
	and 
	$$\left\{j\mid 1\leq j\leq \floor*{\frac{m(r-1)}{r}}\right\} = \left\{j\mid 1\leq j\leq m - \ceil*{\frac{m}{r}}\right\} \subseteq G(Q_1).$$
	Note that  $\floor*{\frac{m\left(r-r + \floor*{\frac{r}{m}}+1\right)}{r}} = \floor*{\frac{m\ceil*{\frac{r}{m}}}{r}} \geq 1$.
	Combining Theorem \ref{theorem:puregap_dim} and Proposition \ref{proposition:puregap_1}, we conclude that $D$ is an effective non-special divisor of degree $g$.
\end{proof}

The support of divisors presented in the above proposition does not contain $Q_0$. 
We now construct effective non-special divisors of degree $g$, whose support includes $Q_0$, splitting the discussion into two cases: $m>r$ and $m<r$.

\begin{proposition}\label{proposition:explicit_non-special2}
	Suppose that $m>r$. Then the divisor
	\begin{equation}\label{equation:non-special_divisor2}
		D = (r-1)Q_0 + \sum_{i=1}^{r-1} \left(\floor*{\frac{m(r-i)}{r}}-1\right)Q_i 
	\end{equation}
	is an effective non-special divisor of degree $g$.
\end{proposition}
\begin{proof}
	From Proposition \ref{proposition:Weierstrass_Q_0} together with $m>r$, we have $\{1,2,\dots,r-1\}\subseteq G(Q_0)$.
	Let integers $1\leq j\leq m-1-\floor*{\frac{m}{r}}$ and $1\leq k\leq r-1$ satisfy $mk - rj = r-1$. 
	For each $1\leq i\leq r-1$, let
	$$t_i = \left\{\begin{array}{cc}
		m-\ceil*{\frac{(k+i)m}{r}}+j,~&\text{ if } k+i \neq r,\\
		m-\ceil*{\frac{(k+i)m}{r}}+j-1,~&\text{ if } k+i = r.\\
	\end{array}\right.$$ 
	Since $mk = r(j+1)-1$, we have
	$$\ceil*{\frac{(k+i)m}{r}} = \ceil*{\frac{r(j+1) + im -1}{r}}  = j+1 + \ceil*{\frac{im-1}{r}}.$$
	Substitution yields
	$$t_i = \left\{\begin{array}{cc}
		m - 1 - \ceil*{\frac{im-1}{r}},~&\text{if}~ k+i \neq r,\\
		m - 1 - \ceil*{\frac{im-1}{r}} - 1,~&\text{if}~ k+i = r.\\
	\end{array}\right.$$ 
	Note that
	$$
	\ceil*{\frac{im-1}{r}} =
	\left\{\begin{array}{cc}
	\ceil*{\frac{im}{r}}-1 & im\equiv 1 \pmod r,\\
	\ceil*{\frac{im}{r}} & im\not\equiv 1 \pmod r.
	\end{array}\right.
	$$
	Since $mk \equiv -1 \pmod r$, then $im\equiv 1 \pmod r$ if and only if $m(i+k)\equiv 0\pmod r$, which is equivalent to $i+k = r$.
	As a consequence, for all $1\leq i\leq r-1$,
	$$t_i = m - 1 - \ceil*{\frac{im}{r}} = \floor*{\frac{m(r-i)}{r}} - 1.$$
	Since $\floor*{\frac{m\left(r-r +1\right)}{r}} -1 = \floor*{\frac{m}{r}} -1 \geq 0$, we have $t_i\geq 0$ for all $1\leq i\leq r-1$.
	Let $1\leq u \leq r-1-\floor*{\frac{r}{m}}$ denote the largest integer with $t_u\geq 1$.
	By Propositions \ref{proposition:subpuregap} and \ref{proposition:puregap_2}, we get that
	\begin{align*}
		\left\{(mk-rj,a_1,\dots,a_{u})\mid 1\leq a_i\leq  t_i \text{ for } 1\leq i\leq u\right\}\subseteq G_0(Q_0,Q_1,\dots,Q_u).
	\end{align*}
	Using Lemma \ref{lemma:floor_sum} (i), we compute the degree
	\begin{align*}
	&\deg \left((r-1)Q_0 + \sum_{i=1}^{u} \left(\floor*{\frac{m(r-i)}{r}}-1\right)Q_i\right) \\
	=~& r-1 + \sum_{i=1}^{r-1}\left(m-2-\floor*{\frac{im}{r}}\right) \\
	=~& r-1 + (m-2)(r-1) - \frac{mr-m-r+1}{2}\\
	=~& \frac{(m-1)(r-1)}{2} = g.
	\end{align*}
	An application of Theorem \ref{theorem:puregap_dim} shows that
	$$D =  (r-1)Q_0 + \sum_{i=1}^{u} \left(\floor*{\frac{m(r-i)}{r}}-1\right)Q_i$$
	is an effective non-special divisor of degree $g$.
\end{proof}

\begin{proposition}\label{proposition:explicit_non-special3}
	Suppose that $m<r$. Then the divisor
	\begin{equation}\label{equation:non-special_divisor3}
	D = (m-1)Q_0 + \sum_{i=1}^{r-2} \floor*{\frac{m(r-i-1)}{r}}Q_i
	\end{equation}
	is an effective non-special divisor of degree $g$.
\end{proposition}
\begin{proof}
	By Proposition \ref{proposition:Weierstrass_Q_0} and $m<r$, we have $\{1,2,\dots,m-1\}\subseteq G(Q_0)$.
	Let integers $1\leq j\leq m-1-\floor*{\frac{m}{r}}$ and $1\leq k\leq r-1$ satisfy $mk - rj = m-1$. 
	For each $1\leq i\leq r-2$, define
	$$t_{i} = \left\{\begin{array}{cc}
		m-\ceil*{\frac{(k+i)m}{r}}+j,~&\text{if}~ k+i \neq r,\\
		m-\ceil*{\frac{(k+i)m}{r}}+j-1,~&\text{if}~ k+i = r,\\
	\end{array}\right.$$ 
	The identity $mk = rj+m-1$ yields
	$$\ceil*{\frac{(k+i)m}{r}} = \ceil*{\frac{(k+i)m}{r}} = \ceil*{\frac{rj + (i+1)m -1}{r}}  = j + \ceil*{\frac{(i+1)m-1}{r}}.$$
	Substituting this expression, 
	$$t_i = \left\{\begin{array}{cc}
		m  - \ceil*{\frac{(i+1)m-1}{r}},~&\text{if}~ k+i \neq r,\\
		m - \ceil*{\frac{(i+1)m-1}{r}} - 1,~&\text{if}~ k+i = r.\\
	\end{array}\right.$$ 
	Note that 
	$$\ceil*{\frac{(i+1)m-1}{r}}=
	\left\{\begin{array}{cc}
	\ceil*{\frac{(i+1)m}{r}}-1 & (i+1)m\equiv 1 \pmod r,\\
	\ceil*{\frac{(i+1)m}{r}} & (i+1)m\not\equiv 1 \pmod r.
	\end{array}\right.$$
	Since $m(k-1) \equiv -1 \pmod r$, then $(i+1)m\equiv 1 \pmod r$ if and only if $m(i+k)\equiv 0\pmod r$, which is equivalent to $i+k = r$.
	Thus for each $1\leq i\leq r-1$, we have
	$$t_{i} = m  - \ceil*{\frac{(i+1)m}{r}} =\floor*{\frac{m(r-i-1)}{r}}.$$
	Since $\floor*{\frac{(r-r+2-1)m}{r}} = \floor*{\frac{m}{r}} = 0$, we have $t_i\geq 0$ for all $1\leq i\leq r-2$.
	Let $1\leq u\leq r-2$ denote the largest integer satisfying $t_u\geq 1$. By Proposition \ref{proposition:puregap_2}, we get
	\begin{align*}
		\left\{(mk-rj,a_1,\dots,a_{u})\mid 1\leq a_i\leq  t_i \text{ for } 1\leq i\leq u\right\}\subseteq G_0(Q_0,Q_1,\dots,Q_u).
	\end{align*}
	Using Lemma \ref{lemma:floor_sum} (i), we compute the degree
	\begin{align*}
		&\deg \left((m-1)Q_0 + \sum_{i=1}^{u} \floor*{\frac{m(r-i-1)}{r}}Q_i\right) \\
        =~& m-1 +\sum_{i=1}^{r-2}\floor*{\frac{m(r-i-1)}{r}} = m-1 +\sum_{i=1}^{r-2}\left(m-1-\floor*{\frac{(i+1)m}{r}}\right)\\
		=~& (m-1)(r-1) - \sum_{i=1}^{r-2}\floor*{\frac{(i+1)m}{r}} - \floor*{\frac{m}{r}} + \floor*{\frac{m}{r}}\\
		=~& (m-1)(r-1) - \sum_{i=1}^{r-1}\floor*{\frac{im}{r}} = (m-1)(r-1) - \frac{mr-m-r+1}{2}\\
		=~& \frac{(m-1)(r-1)}{2} = g.
	\end{align*}
	It follows from Theorem \ref{theorem:puregap_dim} that
    $$D = (m-1)Q_0 + \sum_{i=1}^{u} \floor*{\frac{m(r-i-1)}{r}}Q_i$$ is an effective non-special divisor of degree $g$.
\end{proof}

\begin{remark}
	Let $D_1$ be a divisor of form \eqref{equation:non-special_divisor1}, let $D_2$ be a divisor of form \eqref{equation:non-special_divisor2}, and let $D_3$ be a divisor of form \eqref{equation:non-special_divisor3}.
	If $m>r$, then 
	\begin{align*}
		D_1 - D_2 = \sum_{i=1}^{r-1}Q_i - (r-1)Q_0.
	\end{align*}
	As $r-1\notin H(Q_0)$, we conclude that $D_1$ and $D_2$ are not equivalent. If $m<r$, then $r-\floor*{\frac{r}{m}}-1\leq r-2$.
	Note that $\floor*{\frac{m(r-i)}{r}}\geq \floor*{\frac{m(r-i-1)}{r}}\geq 0$ for $1\leq i\leq r-2$. We have
	\begin{align*}
		D_1 - D_3 = \sum_{i=1}^{r-2}a_i Q_i - (m-1)Q_0,
	\end{align*}
	where $a_i\geq 0$ for $1\leq i\leq r-2$. Since $m-1\notin H(Q_0)$, we conclude that $D_1$ and $D_2$ are not equivalent.
	Therefore, the divisor from Proposition \ref{proposition:explicit_non-special1} is inequivalent to the one from Proposition \ref{proposition:explicit_non-special2}, and neither is equivalent to the divisor given in Proposition \ref{proposition:explicit_non-special3}.
\end{remark}

\section{Constructing LCPs of Codes and LCD codes on Kummer Extensions}\label{section5}
In this section, we develop a general framework for constructing LCPs of AG codes on a Kummer extension.
For an LCP of AG codes $(C_\mathcal{L}(D,G),C_\mathcal{L}(D,H))$ constructed via a canonical divisor, its security is determined by the minimum distance of $C_\mathcal{L}(D,G)$. Moreover, we derive a sufficient condition for $C_\mathcal{L}(D,G)$ to be an LCD code.
Then we apply this result to the function field of the GK curve and the Kummer extension with the same multiplicities. In particular, we construct LCPs of codes on the function field of the GK curve and LCD codes on a Kummer extension given by a linearized polynomial.

Consider the Kummer extension $F = \mathbb{F}_q(x,y)/\mathbb{F}_q(x)$ defined by
$$y^m = \alpha\cdot\prod_{i=1}^{r}p_i(x)^{\lambda_i},$$
where $\alpha\in\mathbb{F}_q$, $p_1(x),\dots,p_r(x)\in\mathbb{F}_q[x]$ are pairwise distinct monic irreducible polynomials, $\lambda_i\geq 1$ for $1\leq i\leq r$.
Let $\lambda_0:= \sum_{i=1}^{r}\lambda_i$. Let $P_i\in\mathbb{P}_{\mathbb{F}_q(x)}$ be the zero of $p_i(x)$ and let $P_0\in\mathbb{P}_{\mathbb{F}_q(x)}$ be the pole of $x$.
For $0\leq i\leq  r$, recall that $r_i:=\gcd(m,\lambda_i)$ and 
$$D_i:=\frac{r_i}{m}\operatorname{Con}_{F/\mathbb{F}_q(x)}(P_i) = \sum_{Q\in\mathbb{P}_F, Q|P} Q.$$
Let $n$ be an integer with $r_0\mid n$. Suppose that there exists a function $h\in F$ with $(h) = \sum_{i=1}^n R_i - \frac{n}{r_0}D_0$, where $R_1,R_2,\dots,R_n$ are $n$ distinct rational places of $F$ with $\operatorname{supp}(\sum_{i=1}^n R_i)\cap \operatorname{supp}(\sum_{i=0}^{r}D_i) = \varnothing$.
Our construction for LCPs of AG codes is given below.
\begin{theorem}\label{theorem:LCP_construction}
	Suppose that $\sum_{i=0}^{r}a_iD_i$ and $\sum_{i=0}^r b_iD_i$ are two non-special divisors of degree $g-1$ and $I\subsetneqq \{1,2,\dots,r\}$. Let $D = (h)_0 = \sum_{i=1}^n R_i$. 
	Suppose that $s$ is an integer satisfying that
	\begin{align}
		&\max_{1\leq i\leq r}\frac{(a_i-b_i)r_i}{\lambda_i}\leq s \leq \frac{(b_0-a_0)r_0 + n}{\lambda_0}~~\text{ and } \label{equation:LCPs_condition1}\\
		&\frac{g-1+\sum_{i\in I\cup\{0\}}(b_i-a_i)r_i}{\lambda_0-\sum_{i\in I}\lambda_i}< s < \frac{n-g+1 + \sum_{i\in I\cup\{0\}}(b_i-a_i)r_i}{\lambda_0-\sum_{i\in I}\lambda_i}. \label{equation:LCPs_condition2}
	\end{align} 
	Define the divisors
	\begin{align*}
		&G = \sum_{i=1,i\notin I}^r a_iD_i + \sum_{i\in I} (s\lambda_i/r_i + b_i)D_i + (n/r_0 - s\lambda_0/r_0+b_0)D_0~~ \text{and}\\
		&H = \sum_{i=1,i\notin I}^r(s\lambda_i/r_i + b_i)D_i + \sum_{i\in I} a_iD_i + a_0D_0.
	\end{align*}
	Set $k=s(\lambda_0-\sum_{i\in I}\lambda_i)-\sum_{i\in I\cup\{0\}}(b_i-a_i)r_i$. 
	Then $(C_{\mathcal{L}}(D,G),C_{\mathcal{L}}(D,H))$ is an LCP of AG codes with parameters 
	\begin{align*}
	[n, n-k,\geq k-g+1 ] \text{ and } [n,k,\geq n-k -g+1],
	\end{align*}
	respectively. 
\end{theorem}
\begin{proof}
	It is obvious that $\operatorname{supp}D\cap\operatorname{supp}G = \operatorname{supp}D\cap\operatorname{supp}H = \varnothing$. Direct degree computation gives
	\begin{align*}
	&\deg G = g-1 + \sum_{i\in I\cup\{0\}}(b_i-a_i)r_i - s\left(\lambda_0-\sum_{i\in I}\lambda_i\right)+ n\text{ and  }\\
	&\deg H = g-1 -\sum_{i\in I\cup\{0\}}(b_i-a_i)r_i + s\left(\lambda_0-\sum_{i\in I}\lambda_i\right).
	\end{align*}
	From \eqref{equation:LCPs_condition2}, we obtain that $2g-2<\deg G, \deg H < n$ and $\deg (G+H) = 2g-2+n$. On the other hand, by \eqref{equation:LCPs_condition1}, we have that 
	\begin{align*}
		&\gcd(G,H) = \sum_{i=0}^r a_i D_i \text{ and }\\
		&\lmd(G,H) = \sum_{i=0}^{r} (s\lambda_i/r_i + b_i)D_i + (b_0 + n/r_0 - s\lambda_0)D_0.
	\end{align*}
	Thus $\gcd(G,H)$ is a non-special divisor of degree $g-1$. 
	By Theorem \ref{theorem:canonical_non-special}, the divisor
	\begin{align*}
		\lmd(G,H) - D &= \sum_{i=0}^{r} (s\lambda_i/r_i + b_i)D_i + (n/r_0 - s\lambda_0 + b_0)D_0 -D\\
					  &= \sum_{i=0}^{r} b_iD_i + s\sum_{i=0}^{r}\frac{\lambda_i}{r_i}D_i - s\frac{\lambda_0}{r_0}D_0 + \frac{n}{r_0}D_0 - D\\
					  &= \sum_{i=0}^{r} b_iD_i + (y^s) - (h) \sim \sum_{i=0}^{r} b_iD_i
	\end{align*}
	is also a non-special divisor of degree $g-1$. Thus $(C_{\mathcal{L}}(D,G),C_{\mathcal{L}}(D,H))$ is an LCP of AG codes by Theorem \ref{theorem:LCP}. 
	The parameters of $C_{\mathcal{L}}(D,G)$ and $C_{\mathcal{L}}(D,H)$ are obtained from Theorem \ref{theorem:AG_codes}.
\end{proof}

Using a canonical divisor generated by $D_0,D_1,\dots,D_r$, we establish the following corollary, which generalizes the main results in \cite[Theorems 9 and 10]{castellanosLinearComplementary2025}.
\begin{corollary}\label{theorem:LCP_and_LCD}
	Suppose that $\sum_{i=0}^{r}a_iD_i$ is a non-special divisor of degree $g-1$ and $I\subsetneqq \{1,2,\dots,r\}$.
	Let $\omega$ be a Weil differential such that $(\omega) = \sum_{i=0}^r b_iD_i$ and $D = (h)_0 = \sum_{i=1}^n R_i$.
	Suppose that $s$ is an integer satisfying that
	\begin{align}
		&\max_{1\leq i\leq r}\frac{(2a_i-b_i)r_i}{\lambda_i}\leq s \leq \frac{(b_0-2a_0)r_0 + n}{\lambda_0}~~\text{ and } \label{equation:s_condition1}\\
		&\frac{g-1+\sum_{i\in I\cup\{0\}}(b_i-2a_i)r_i}{\lambda_0-\sum_{i\in I}\lambda_i}< s < \frac{n-g+1 + \sum_{i\in I\cup\{0\}}(b_i-2a_i)r_i}{\lambda_0-\sum_{i\in I}\lambda_i}. \label{equation:s_condition2}
	\end{align} 
	Define the divisors
	\begin{align*}
		&G = \sum_{i=1,i\notin I}^r a_iD_i + \sum_{i\in I} (s\lambda_i/r_i + b_i - a_i)D_i + (n/r_0 - s\lambda_0/r_0+b_0-a_0)D_0~~ \text{and}\\
		&H = \sum_{i=1,i\notin I}^r(s\lambda_i/r_i + b_i - a_i)D_i + \sum_{i\in I} a_iD_i + a_0D_0.
	\end{align*}
	Set $k=s(\lambda_0-\sum_{i\in I}\lambda_i)-\sum_{i\in I\cup\{0\}}(b_i-2a_i)r_i$. 
	Then $(C_{\mathcal{L}}(D,G),C_{\mathcal{L}}(D,H))$ is an LCP of AG codes with parameters 
	\begin{align*}
	[n, n-k,\geq k-g+1 ] \text{ and } [n,k,\geq n-k -g+1],
	\end{align*}
	respectively. The security parameter of $(C_{\mathcal{L}}(D,G),C_{\mathcal{L}}(D,H))$ is $d(C_{\mathcal{L}}(D,G))$.
	Moreover, the AG code $C_{\mathcal{L}}(D,G)$ is an LCD code if $\operatorname{res}_{R_i}\left(\frac{y^s}{h}\omega\right) = \operatorname{res}_{R_j}\left(\frac{y^s}{h}\omega\right)$ for all $1\leq i<j\leq n$.
\end{corollary}
\begin{proof}
	By Theorem \ref{theorem:canonical_non-special}, we have that $\sum_{i=0}^{r}(b_i-a_i)D_i$ is a non-special divisor of degree $g-1$.
	An application of Theorem \ref{theorem:LCP_construction} together with Theorem \ref{theorem:LCD} yields that $(C_{\mathcal{L}}(D,G),C_{\mathcal{L}}(D,H))$ is an LCP of AG codes.
	On the other hand, we have 
	\begin{align*}
		G + H &= \sum_{i=0}^r b_i D_i  + s\sum_{i=0}^{r}\frac{\lambda_i}{r_i}D_i - s\frac{\lambda_0}{r_0}D_0 + \frac{n}{r_0}D_0\\
			  &= (\omega) + (y^s) - (h) + D.
	\end{align*}
	Then we get $G = D-H + \left(\frac{y^s}{h}\omega\right)$ and $v_{R_i}\left(\frac{y^s}{h}\omega\right) = -1$ for all $1\leq i\leq n$. Thus
	$$C_{\mathcal{L}}(D,H)^\bot = \mathbf{a}\cdot C_\mathcal{L}(D,G)\sim  C_\mathcal{L}(D,G),$$
	where $\mathbf{a} = \left(\operatorname{res}_{R_1}\left(\frac{y^s}{h}\omega\right),\dots,\operatorname{res}_{R_n}\left(\frac{y^s}{h}\omega\right)\right)$.
	This implies that $d(C_{\mathcal{L}}(D,G))$ is the security parameter of $(C_{\mathcal{L}}(D,G),C_{\mathcal{L}}(D,H))$.
	Moreover, it follows from Theorem \ref{theorem:LCD} that $C_{\mathcal{L}}(D,G)$ is an LCD code if $\operatorname{res}_{R_i}\left(\frac{y^s}{h}\omega\right) = \operatorname{res}_{R_j}\left(\frac{y^s}{h}\omega\right)$ for all $1\leq i<j\leq n$.
\end{proof}

\subsection{LCPs of Codes on the GK Curve}
In this subsection, consider the function field of the GK curve $F = \mathbb{F}_{q^6}(x,y)/\mathbb{F}_{q^6}(x)$ defined by \eqref{equation:GKcurve}:
$$y^{q^3+1} = (x^q+x)\left(\frac{x^{q^2}-x}{x^q+x}\right)^{q+1} = \left(\prod_{i=1}^q(x-\alpha_i)\right) \left(\prod_{i=1}^{q^2-q}(x-\beta_i)^{q+1}\right),$$
where $\alpha_1,\dots,\alpha_q,\beta_1,\dots,\beta_{q^2-q}\in \mathbb{F}_{q^2}$.
The genus of $F$ is $g = \frac{1}{2}(q^3+1)(q^2-2)+1$ and has $N = q^8-q^6+q^5+1$ rational places.
By \eqref{equation:dx_divisor}, a canonical divisor takes the form
\begin{equation}\label{equation:GK_dx}
	(dx) = \sum_{i=1}^{q}q^3 Q_i + \sum_{i=1}^{q^2-q-1}(q^2-q)D_i - (q^3+2)Q_0,
\end{equation}
where $Q_0 = \frac{1}{q^3+1}(x)_\infty$, $Q_i = \frac{1}{q^3+1}(x-\alpha_i)_0$ for $1\leq i\leq q$, and $D_i = \frac{1}{q^2-q+1}(x-\beta_i)_0$ for $1\leq i\leq q^2-q$.
For each $1\leq j\leq q^2-q$, we rewrite
$$y^{q+1} = \frac{x^q+x}{(x^q+x)^{{q+1}}}\left(\prod_{i=1,i\neq j}^{q^2-q}(x-\beta_i)\right)^{q+1}(x-\beta_j)^{q+1}.$$
Since $\beta_j^q + \beta_j\in\mathbb{F}_q$, there exists $\gamma\in \mathbb{F}_{q^2}$ such that $\gamma^{q+1} = \beta_j^q + \beta_j$. Then 
$$\frac{\beta_j^q+\beta_j}{(\beta_j^q+\beta_j)^{{q+1}}}\left(\prod_{i=1,i\neq j}^{q^2-q}(\beta_j-\beta_i)\right)^{q+1}$$
is an $(q+1)$-th power of an element in $\mathbb{F}_{q^6}$. Thus $D_j$ exactly consists of $q+1$ distinct rational places $Q_{j,1},\dots,Q_{j,q+1}$.
We obtain that the number of rational places except $\{Q_i\mid 0\leq i\leq q\}\cup\{Q_{j,k}\mid 1\leq j\leq q^2-q,\,1\leq k\leq q+1\}$ is 
$$q^8-q^6 + q^5 + 1 - (q+1) - (q^2-q)(q+1) = q^8-q^6+q^5-q^3 = (q^3+1)(q^5-q^3).$$
Therefore, there exist $q^5-q^3$ elements $\gamma_1,\gamma_2,\dots,\gamma_{q^5-q^3}\in\mathbb{F}_{q^6}$ such that the zero of $x-\gamma_i$ in $\mathbb{P}_{\mathbb{F}_{q^6}(x)}$ splits completely in $F/\mathbb{F}_{q^6}(x)$ for each $1\leq i\leq q^5-q^3$.
Our explicit construction of LCPs of AG codes on the function field of the GK curve is presented below.
\begin{proposition}\label{proposition:GK_LCP}
	Let $D = \left(\prod_{i=1}^{q^5-q^3}(x-\gamma_i)\right)_0$ and $n = (q^3+1)(q^5-q^3)$. Suppose that $s$ is an integer satisfying that
		$$q^3-2q^2+2q-2\leq s < \frac{n}{q^3} - 1 - \frac{1}{2}(1 + 1/q^3)(q^2-2).$$
	Consider the divisors
	\begin{align*}
		&G = \sum_{i=1}^q(i(q^2-q+1)-1)Q_i + \sum_{i=1}^{q^2-q-1} iD_i  + (n - sq^3-q^3-1)Q_0~~ \text{and}\\
		&H = \sum_{i=1}^q(s + q^3-i(q^2-q+1)+1)Q_i + \sum_{i=1}^{q^2-q-1} (s+q^2-q-i)D_i + (s+q^2-q)D_{q^2-q} -Q_0.
	\end{align*}
	Then $(C_{\mathcal{L}}(D,G),C_{\mathcal{L}}(D,H))$ is an LCP of AG codes with parameters 
	\begin{align*}
	&[n, n-q^3(s+1),\geq q^3(s+1)-(q^3+1)(q^2-2)/2 ] \text{ and } \\
	&[n,q^3(s+1),\geq n-q^3(s+1)-(q^3+1)(q^2-2)/2],
	\end{align*}
	respectively, and the security parameter of  $(C_{\mathcal{L}}(D,G),C_{\mathcal{L}}(D,H))$ is $d(C_{\mathcal{L}}(D,H))$.
\end{proposition}
\begin{proof}
	By Proposition \ref{proposition:GK_non-special}, the divisor
	$$\sum_{i=0}^{q}(i(q^2-q+1)-1)Q_i + \sum_{i=1}^{q^2-q-1}iD_i$$
	is a non-special divisor of degree $g-1$. Let $(\omega) = (dx)$ be the canonical divisor $\eqref{equation:GK_dx}$ and $I = \varnothing$ in Corollary \ref{theorem:LCP_and_LCD}. It suffices to verify that $s$ satisfies the conditions \eqref{equation:s_condition1} and \eqref{equation:s_condition2}.
	One computes
	$$\max_{1\leq i\leq q} \left\{2i(q^2-q+1)-2-q^3\right\} = 2q(q^2-q+1)-2-q^3 = q^3-2q^2+2q-2,$$
	$$\max_{0\leq i\leq q^2-q-1} \left\{2i-(q^2-q)\right\} = q^2-q-2< q^3-2q^2+2q-2,$$
	and
	$$\frac{g-1 + (-q^3-2+2)}{q^3} = \frac{1}{2}(1+1/q^3)(q^2-2)<q^3-2q^2+2q-2.$$
	We also have that 
	$$\frac{(-q^3-2+2) + n}{q^3} > \frac{n - g+1 + (-q^3-2+2)}{q^3} = \frac{n}{q^3} - 1 - \frac{1}{2}(1 + 1/q^3)(q^2-2).$$
	Thus $s$ satisfies the conditions \eqref{equation:s_condition1} and \eqref{equation:s_condition2}.
\end{proof}
\begin{example}
	Let $q = 2$. We have the function field $F = \mathbb{F}_{2^6}(x,y)$ defined by 
	$$y^9 = (x^2+x)(x^2+x+1)^3.$$
	The genus of $F$ is $g = 10$. Suppose that $s$ is an integer satisfying that $2\leq s \leq 24$. Take
	\begin{align*}
		&G = 2Q_1 + 5Q_2 + D_1 + (207-8s)Q_0,\\
		&H = (s+6)Q_1 + (s+3)Q_2 + (s+1)D_1 + (s+2)D_2 - Q_0, \text{ and }\\
		&D = \sum_{P\in\mathbb{P}_F^1\setminus \operatorname{supp} H} P.
	\end{align*}
	By Proposition \ref{proposition:GK_LCP}, we obtain that $(C_{\mathcal{L}}(D,G),C_{\mathcal{L}}(D,H))$ is an LCP of AG codes with parameters 
	\begin{align*}
	[216, 208-8s,\geq 8s-1 ] \text{ and } [216, 8s+8,\geq 199-8s],
	\end{align*}
	respectively, and the security parameter of  $(C_{\mathcal{L}}(D,G),C_{\mathcal{L}}(D,H))$ is $\geq 8s-1$.
\end{example}

\subsection{LCPs of Codes on Kummer Extensions with the Same Multiplicities}
In this subsection, let integers $r\ge2$, $m\ge2$ satisfy $\gcd(m,q)=\gcd(m,r)=1$. We treat the Kummer extension $F = \mathbb{F}_q(x,y)/\mathbb{F}_q(x)$ defined by  
$$y^m = \prod_{i=1}^{r}(x-\alpha_i),$$
where $\alpha_1,\dots,\alpha_r\in\mathbb{F}_q$ are pairwise distinct. The genus of $F$ is $g = (m-1)(r-1)/2$.
Let $Q_0\in\mathbb{P}_{F}$ be the only pole of $x$, and let $Q_i\in\mathbb{P}_{F}$ be the only zero of $x-\alpha_i$ for $1\leq i\leq r$.
From \eqref{equation:dx_divisor}, a canonical divisor reads
\begin{equation}\label{equation:same_dx}
(dx) = \sum_{i=1}^{r}(m-1) Q_i -(m+1) Q_0.
\end{equation}
Assume that $\beta_1,\dots,\beta_t \in \mathbb{F}_q$ are chosen such that the zero of $x-\beta_i$ in $\mathbb{P}_{\mathbb{F}_q(x)}$ completely splits in the extension $F/\mathbb{F}_q(x)$ for $1\leq i\leq t$.
Set $h = \prod_{i=1}^t(x-\beta_i)$. In the subsequent three propositions, we explicitly construct several families of LCPs of AG codes on $F$.
\begin{proposition}\label{proposition:same_LCP1}
	Let $D = (h)_0$ and $n = mt$. 
	Let integer $s$ satisfy $m-1\leq s < \frac{n-g-m+2}{r}$. Define the divisors
	$$G = \sum_{i=1}^{r - \floor*{\frac{r}{m}} -1} \floor*{\frac{m(r-i)}{r}} Q_i + (n-rs-m)Q_0, \text{ and }$$
	$$H = \sum_{i=1}^{r - \floor*{\frac{r}{m}} -1} \left(s+\floor*{\frac{im}{r}}\right)Q_i + \sum_{i=r-\floor*{\frac{r}{m}}}^{r}(s+m-1)Q_i - Q_0.$$
	Then the pair $(C_\mathcal{L}(D,G), C_\mathcal{L}(D,H))$ is an LCP of AG codes with parameters 
	$$[n,n-(rs+m-1), \geq rs+m-g] \text{ and } [n,rs+m-1,\geq n-rs-m-g+2],$$
	respectively. Moreover, the security parameter of  $(C_\mathcal{L}(D,G),C_\mathcal{L}(D,H))$ is $d(C_\mathcal{L}(D,G))$.
\end{proposition}
\begin{proof}
	By Proposition \ref{proposition:explicit_non-special1} and Lemma \ref{lemma:non-special_g_to_g-1}, the divisor 
	$$\sum_{i=\floor*{\frac{r}{m}}+1}^{r-1} \floor*{\frac{im}{r}} Q_i - Q_0$$
	is a non-special divisor of degree $g-1$. Let $(\omega) = (dx)$ be the canonical divisor \eqref{equation:same_dx} and $I = \varnothing$ in Corollary \ref{theorem:LCP_and_LCD}.
	Using the identity $\floor*{\frac{im}{r}} + \floor*{\frac{m(r-i)}{r}} = m-1$ for all $1\leq i\leq r-1$, we only need to confirm that $s$ satisfies the conditions \eqref{equation:s_condition1} and \eqref{equation:s_condition2}.
	One calculates
	$$ \max_{1\leq i\leq r - \floor*{\frac{r}{m}} -1}\left\{2\floor*{\frac{m(r-i)}{r}} - (m-1)\right\} =  \max_{1\leq i\leq r - \floor*{\frac{r}{m}} -1}\left\{m-2\ceil*{\frac{im}{r}}+1\right\}\leq m-1,$$
	and $$\frac{g - 1 + (-m-1+2)}{r} = \frac{1}{2}(m-1)(1-1/r)-\frac{m}{r}<m-1.$$
	We also have that 
	$$\frac{n-g+1-m-1+2}{r} = \frac{n-g-m+2}{r} \leq \frac{-m-1+2 + n}{r}.$$
	By Corollary \ref{theorem:LCP_and_LCD}, the pair $(C_\mathcal{L}(D,G), C_\mathcal{L}(D,H))$ is an LCP of AG codes with parameters 
	$$[n,n-(rs+m-1), \geq rs+m-g] \text{ and } [n,rs+m-1,\geq n-rs-m-g+2],$$
	respectively. Moreover, the security parameter of  $(C_\mathcal{L}(D,G),C_\mathcal{L}(D,H))$ is $d(C_\mathcal{L}(D,G))$.
\end{proof}

\begin{proposition}\label{proposition:same_LCP2}
	Suppose that $m>r$, $D = (h)_0$ and $n = mt$. 
	Let integer $s$ satisfy $m-1\leq s<\frac{n-g-2r+m+2}{r}$. Define the divisors
	$$G = \sum_{i=1}^{r-1} \left(\floor*{\frac{m(r-i)}{r}}-1\right) Q_i +(m-1) Q_r + (n-rs-r)Q_0, \text{ and }$$
	$$H = \sum_{i=1}^{r-1}\left(s + \floor*{\frac{im}{r}}+1\right)Q_i + sQ_r + (r-1-m)Q_0.$$
	Then the pair $(C_\mathcal{L}(D,G),C_\mathcal{L}(D,H))$ is an LCP of AG codes with parameters
	$$[n,n-(s+2)r+m+1, \geq (s+2)r-m-g] \text{ and } [n,(s+2)r-m-1,\geq n-(s+2)r+m-g+2],$$
	respectively. Moreover, the security parameter of  $(C_\mathcal{L}(D,G),C_\mathcal{L}(D,H))$ is $d(C_\mathcal{L}(D,H))$.
\end{proposition}
\begin{proof}
	By Proposition \ref{proposition:explicit_non-special2} and Lemma \ref{lemma:non-special_g_to_g-1}, the divisor 
	\begin{align*}
		& \sum_{i=1}^{r-1} \left(\floor*{\frac{m(r-i)}{r}}-1\right) Q_i +(m-1) Q_r + (r-1-m)Q_0\\
		&\sim \sum_{i=1}^{r-1} \left(\floor*{\frac{m(r-i)}{r}}-1\right) Q_i - Q_r + (r-1)Q_0
	\end{align*}
	is a non-special divisor of degree $g-1$. Let $(\omega) = (dx)$ be the canonical divisor \eqref{equation:same_dx} and $I = \varnothing$ in Corollary \ref{theorem:LCP_and_LCD}.
	Note that $\floor*{\frac{im}{r}} + \floor*{\frac{m(r-i)}{r}} = m-1$ for $1\leq i\leq r-1$. It suffices to verify that $s$ satisfies the conditions \eqref{equation:s_condition1} and \eqref{equation:s_condition2}.
	We have that 
	$$ \max_{1\leq i\leq r-1}\left\{2\floor*{\frac{m(r-i)}{r}}-2- (m-1)\right\} =  \max_{1\leq i\leq r-1}\left\{m-2\ceil*{\frac{im}{r}}-1\right\}\leq m-1,$$
	and $$\frac{g - 1 + (-m-1-2r+2+2m)}{r} = \frac{1}{2}(m-1)(1-1/r)-2 + \frac{m}{r}<m-1.$$
	We also have that 
	$$\frac{n-g+1-m-1-2r+2+2m}{r} = \frac{n-g+m-2r+2}{r} \leq \frac{-1-m-2r+2+2m+n}{r}.$$
	By Corollary \ref{theorem:LCP_and_LCD}, the pair $(C_\mathcal{L}(D,G), C_\mathcal{L}(D,H))$ is an LCP of AG codes with parameters 
	$$[n,n-(s+2)r+m+1, \geq (s+2)r-m-g] \text{ and } [n,(s+2)r-m-1,\geq n-(s+2)r+m-g+2],$$
	respectively. Moreover, the security parameter of  $(C_\mathcal{L}(D,G),C_\mathcal{L}(D,H))$ is $d(C_\mathcal{L}(D,H))$.
\end{proof}
\begin{proposition}\label{proposition:same_LCP3}
	Suppose that $r>m$, $D = (h)_0$ and $n = mt$.
	Let integer $s$ satisfy $m-1\leq s<\tfrac{n-g-m+2}{r}$. Define the divisors 
	$$G = \sum_{i=1}^{r-2} \floor*{\frac{m(r-i-1)}{r}}Q_i +(m-1)Q_r +(n-rs-m)Q_0, \text{ and }$$
	$$H = \sum_{i=1}^{r-2}\left(s + \floor*{\frac{m(i+1)}{r}}\right)Q_i  + (s+m-1)Q_{r-1} + sQ_r - Q_0.$$
	Then $(C_\mathcal{L}(D,G),C_\mathcal{L}(D,H))$ is an LCP of AG codes with parameters
	$$[n,n-(rs+m-1), \geq rs+m-g] \text{ and } [n,rs+m-1,\geq n-rs-m-g+2],$$
	respectively. Moreover, the security parameter of  $(C_\mathcal{L}(D,G),C_\mathcal{L}(D,H))$ is $d(C_\mathcal{L}(D,H))$. 
\end{proposition}
\begin{proof}
	From Proposition \ref{proposition:explicit_non-special3} and Lemma \ref{lemma:non-special_g_to_g-1}, the divisor 
	\begin{align*}
		&\sum_{i=1}^{r-2} \floor*{\frac{m(r-i-1)}{r}}Q_i +(m-1) Q_r - Q_0\\
		&\sim \sum_{i=1}^{r-2} \floor*{\frac{m(r-i-1)}{r}}Q_i - Q_r + (m-1) Q_0
	\end{align*}
	is a non-special divisor of degree $g-1$. Let $(\omega) = (dx)$ be the canonical divisor \eqref{equation:same_dx} and $I = \varnothing$ in Corollary \ref{theorem:LCP_and_LCD}.
	The identity $\floor*{\frac{m(i+1)}{r}} + \floor*{\frac{m(r-i-1)}{r}} = m-1$ holds for all $1\leq i\leq r-2$, so we only need to verify that $s$ satisfies the conditions \eqref{equation:s_condition1} and \eqref{equation:s_condition2}.
	One computes 
	$$ \max_{1\leq i\leq r-2}\left\{2\floor*{\frac{m(r-i-1)}{r}}- (m-1)\right\} =  \max_{1\leq i\leq r-2}\left\{m-2\ceil*{\frac{(i+1)m}{r}}+1\right\}\leq m-1,$$
	and $$\frac{g - 1 + (-m-1+2)}{r} = \frac{1}{2}(m-1)(1-1/r)-\frac{m}{r}<m-1.$$
	We also have that 
	$$\frac{n-g+1-m-1+2}{r} = \frac{n-g-m+2}{r} \leq \frac{-m-1 +2 + n}{r}.$$
	By Corollary \ref{theorem:LCP_and_LCD}, the pair $(C_\mathcal{L}(D,G), C_\mathcal{L}(D,H))$ is an LCP of AG codes with parameters 
	$$[n,n-(rs+m-1), \geq rs+m-g] \text{ and } [n,rs+m-1,\geq n-rs-m-g+2],$$
	respectively. Moreover, the security parameter of  $(C_\mathcal{L}(D,G),C_\mathcal{L}(D,H))$ is $d(C_\mathcal{L}(D,H))$.
\end{proof}
\begin{example}
	Let $m\mid (q+1)$. Consider the Kummer extension $F = \mathbb{F}_{q^2}(x,y)/\mathbb{F}_{q^2}(x)$ defined by a quotient of the Hermitian curve
	$$y^{m} = x^q+x.$$
	From \cite[Example 6.4.2]{stichtenothAlgebraicFunctionFields2009}, the function field $F$ has $1+q(q+(q-1)m)$ rational places and genus $(q-1)(m-1)/2$. The number of rational places except $\{Q_i \mid 0\leq i\leq q\}$ is $m(q^2-q)$.
	Thus for $\alpha \in \mathbb{F}_{q^2}\setminus \{\beta\in\mathbb{F}_{q^2}\mid \beta^2+\beta = 0\}$, the zero of $x-\alpha$ in $\mathbb{P}_{\mathbb{F}_{q^2}(x)}$ splits completely.
	
	First set $q = 5$, $m=6$, and $D = \left(\frac{x^{25}-x}{x^5+x}\right)_0$. Suppose that $s$ is an integer satisfying that $5\leq s \leq 21$. Take
	\begin{align*}
		&G_1 = 4Q_1 + 3Q_2 + 2Q_3 + Q_4 + (114-5s)Q_0 , \text{ and }\\
		&H_1 = (s+1)Q_1 + (s+2)Q_2 + (s+3)Q_3 + (s+4)Q_4 + (s+5)Q_5 - Q_0.
	\end{align*}
	By Proposition \ref{proposition:same_LCP1}, we obtain that $(C_{\mathcal{L}}(D,G_1),C_{\mathcal{L}}(D,H_1))$ is an LCP of AG codes with parameters 
	\begin{align*}
	[120, 115-5s,\geq 5s-4 ] \text{ and } [120, 5s+5,\geq 107-5s],
	\end{align*}
	respectively, and the security parameter of  $(C_{\mathcal{L}}(D,G_1),C_{\mathcal{L}}(D,H_1))$ is $\geq 5s-4$.
	Suppose that $t$ is an integer satisfying that $5\leq t \leq 21$. Take
	\begin{align*}
		&G_2 = 3Q_1 + 2Q_2 + Q_3 + 5Q_5 + (115-5t)Q_0 , \text{ and }\\
		&H_2 = (t+2)Q_1 + (t+3)Q_2 + (t+4)Q_3 + (t+5)Q_4 + tQ_5 - 2Q_0.
	\end{align*}
	By Proposition \ref{proposition:same_LCP2}, we obtain that $(C_{\mathcal{L}}(D,G_2),C_{\mathcal{L}}(D,H_2))$ is an LCP of AG codes with parameters 
	\begin{align*}
	[120, 117-5t,\geq 5t-6 ] \text{ and } [120, 5t+3,\geq 108-5t],
	\end{align*}
	respectively, and the security parameter of  $(C_{\mathcal{L}}(D,G_2),C_{\mathcal{L}}(D,H_2))$ is $\geq 5t-6$.

	Now let $q = 9$, $m=5$, and $D = \left(\frac{x^{81}-x}{x^9+x}\right)_0$. Suppose that $s$ is an integer satisfying that $4\leq s \leq 38$. Take
	\begin{align*}
		&G = 3Q_1 + 3Q_2 + 2Q_3 + 2Q_4 + Q_5 + Q_6 + 4Q_9 + (365-9s)Q_0 , \text{ and }\\
		&H = (s+1)Q_1 + (s+1)Q_2 + (s+2)Q_3 + (s+2)Q_4 + (s+3)Q_5 + (s+3)Q_6\\
		&~~~~~~~+(s+4)Q_7 + (s+4)Q_8 + sQ_9 - Q_0.
	\end{align*}
	By Proposition \ref{proposition:same_LCP3}, we obtain that $(C_{\mathcal{L}}(D,G),C_{\mathcal{L}}(D,H))$ is an LCP of AG codes with parameters 
	\begin{align*}
	[360, 356-9s,\geq 9s-11] \text{ and } [360, 9s+4,\geq 341-9s],
	\end{align*}
	respectively, and the security parameter of  $(C_{\mathcal{L}}(D,G),C_{\mathcal{L}}(D,H))$ is $\geq 9s-11$.
\end{example}

\subsection{LCD Codes on Kummer Extensions with a Linearized Polynomial}
In this subsection, let $L(x) = \sum_{i=0}^{r} \alpha_i x^{p^i}$ denote a linearized polynomial satisfying $\alpha_0,\alpha_r\neq 0$ with exactly $p^r$ distinct roots in $\mathbb{F}_q$. We investigate the Kummer extension $F = \mathbb{F}_{q}(x,y)/\mathbb{F}_q(x)$ given by 
	$$y^m = L(x),$$
where $\gcd(m,p) = 1$. The genus of $F$ is $g = (m-1)(p^r-1)/2$.
Denote by $Q_0\in\mathbb{P}_{F}$ the unique pole of $x$, and let $Q_1,\dots,Q_{p^r}\in\mathbb{P}_{F}$ be the zeros of $L(x)$.
By \eqref{equation:dx_divisor}, we have a canonical divisor
\begin{equation}\label{equation:linearized_dx}
	(dx) = \sum_{i=1}^{p^r}(m-1) Q_i -(m+1) Q_0.
\end{equation}
Let $t\geq 2m-1$ be a positive integer such that $ t\mid (q-1)$.
Suppose that the polynomial $L(x) - \beta^m$ has $p^r$ distinct roots in $\mathbb{F}_q$ for every root $\beta\in \mathbb{F}_q$ of $y^t-1$. 
Let $D = (y^t - 1)_0$, then $\#\operatorname{supp} D = p^rt$.
We construct LCD AG codes on the Kummer extension with a linearized polynomial as follows.
\begin{proposition}\label{proposition:LCD_linearized}
	Define the divisors
	\begin{align*}
		&G_1 = \sum_{i=1}^{p^r-\floor*{\frac{p^r}{m}}-1}\floor*{\frac{m(p^r-i)}{p^r}}Q_i + (p^rt - p^r(t-m)-m)Q_0,\\
		&G_2 = \sum_{i=1}^{p^r-1}\left(\floor*{\frac{m(p^r-i)}{p^r}}-1\right)Q_i +(m-1) Q_{p^r} + (p^rt - p^r(t-m)- p^r)Q_0,\\
		&G_3 = \sum_{i=1}^{p^r-2}\floor*{\frac{m(p^r-i-1)}{p^r}}Q_i + (m-1) Q_{p^r}+ (p^rt - p^r(t-m) - m)Q_0.
	\end{align*}
	Then $C_{\mathcal{L}}(D,G_1)$ is a $[p^rt, mp^r-m+1,\geq p^r(t-m)+m-g]$-LCD code.
	
	If $m>p^r$, then $C_{\mathcal{L}}(D,G_2)$ is an LCD code with parameters 
	$$[p^rt, mp^r+m+1,\geq p^r(t-m+2)-m-g].$$
	If $m<p^r$, then $C_{\mathcal{L}}(D,G_3)$ is an LCD code with parameters 
	$$[p^rt, (m-2)p^r+m+1,\geq p^r(t-m)+m-g].$$
\end{proposition}
\begin{proof}
	Let $(\omega) = (dx)$ be the canonical divisor \eqref{equation:linearized_dx} and $s = t-m$ in Corollary \ref{theorem:LCP_and_LCD}. Since $my^{m-1}dy = \alpha_0 dx$, we have
	\begin{align*}
		\frac{y^{t-m}}{y^t-1}dx = \frac{m}{\alpha_0}\cdot\frac{y^{t-1}}{y^t-1} dy = \frac{m}{t\alpha_0}\cdot\frac{d(y^t-1)}{y^t-1}.
	\end{align*} 
	Thus $\operatorname{res}_{P}\left(\frac{y^{t-m}}{y^t-1}dx\right) = \frac{m}{t\alpha_0}$ for all $P\in \operatorname{supp} D$. 
	Note that 
	$$ m-1\leq t-m < t - \frac{1}{2}(m-1)(1-1/p^r)-(m-2)/p^r = \frac{p^rt - (m-1)(p^r-1)/2-m+2}{p^r}.$$
	Combining Corollary \ref{theorem:LCP_and_LCD} and Proposition \ref{proposition:same_LCP1}, the AG code $C_{\mathcal{L}}(D,G_1)$ is a $[p^rt, mp^r-m+1,\geq p^r(t-m)+m-g]$-LCD code. 
	The remaining two assertions follow by analogous arguments from Propositions \ref{proposition:same_LCP2} and \ref{proposition:same_LCP3}.
\end{proof}
	Consider the Kummer extension $F = \mathbb{F}_{q^2}(x,y)/\mathbb{F}_{q^2}(x)$ defined by a quotient of the Hermitian curve:
	$$y^m = x^q + x,$$
	where $m\mid (q+1)$. The genus of $F$ is $g = (q-1)(m-1)/2$. From \cite[Example 6.4.2]{stichtenothAlgebraicFunctionFields2009}, 
	let $U\subseteq \mathbb{F}_{q^2}^\ast$ be the subgroup of order $(q-1)m$. Then for $\alpha\in\mathbb{F}_{q^2}$, we have 
	$$\alpha^m \in \mathbb{F}_q \text{ if and only if } \alpha\in U\cup \{0\}.$$
	Thus the equation $T^q + T = \alpha^m$ has $q$ distinct roots in $\mathbb{F}_{q^2}$ if and only if $\alpha\in  U\cup \{0\}$.
	Let $h = \prod_{\alpha\in U}(y-\alpha) = y^{(q-1)m}-1$. We now construct explicit LCD AG codes on the function field of a quotient of the Hermitian function field via Proposition \ref{proposition:LCD_linearized}.
	\begin{corollary}\label{corollary:LCD_Hermitian}
		Suppose that $q>2$ and $D = (h)_0$. Define the divisors 
		\begin{align*}
		&G_1 = \sum_{i=1}^{q-\floor*{\frac{q}{m}}-1}\floor*{\frac{m(q-i)}{q}}Q_i + (qm-m)Q_0,\\
		&G_2 = \sum_{i=1}^{q-1}\left(\floor*{\frac{m(q-i)}{q}}-1\right)Q_i + (m-1)Q_{q} + (qm- q)Q_0,\\
		&G_3 = \sum_{i=1}^{q-2}\floor*{\frac{m(q-i-1)}{q}}Q_i + (m-1)Q_{q}+ (qm-m)Q_0.
	\end{align*}
	Then $C_{\mathcal{L}}(D,G_1)$ is a $[q(q-1)m, mq-m+1,\geq q((q-1)m-m)+m-g]$-LCD code.

	If $m>q$, then $C_{\mathcal{L}}(D,G_2)$ is an LCD code with parameters 
	$$[q(q-1)m, (m-2)q+m+1,\geq q((q-1)m-m+2)-m-g].$$
	If $m<q$, then $C_{\mathcal{L}}(D,G_3)$ is an LCD code with parameters 
	$$[q(q-1)m, m(q-1)+1,\geq q((q-1)m-m) + m-g].$$
	\end{corollary}
	\begin{proof}
		Set $t=(q-1)m$. All assertions follow directly from Proposition \ref{proposition:LCD_linearized}.
	\end{proof}
	\begin{example}
	Consider the Kummer extension $\mathbb{F}_{16}(x,y)/\mathbb{F}_{16}(x)$ defined by $y^{5} = x^4+x$.
	Take $G = 3Q_1 + 2Q_2 + Q_3 + 15Q_0$ and $D = (y^{15}-1)_0$. Then the AG code $C_\mathcal{L}(D,G)$ is a $[60,16,\geq 39]$-LCD code.

	Consider the Kummer extension $\mathbb{F}_{81}(x,y)/\mathbb{F}_{81}(x)$ defined by $y^{6} = x^5+x$.
	Take $G = 3Q_1 + 2Q_2 + Q_3 + 5Q_5 + 25Q_0$ and $D = (y^{24}-1)_0$. Then the AG code $C_\mathcal{L}(D,G)$ is a $[120,27,\geq 84]$-LCD code.

	Consider the Kummer extension $\mathbb{F}_{25}(x,y)/\mathbb{F}_{25}(x)$ defined by $y^{3} = x^5+x$.
	Take $G = 2Q_1 + Q_2 + Q_3 + 12Q_0$ and $D = (y^{12}-1)_0$. Then the AG code $C_\mathcal{L}(D,G)$ is a $[60,13,\geq 44]$-LCD code.
	\end{example}

\section{Conclusion}
In this paper, we studied LCPs of AG codes and LCD AG codes on Kummer extensions. We provided an arithmetic characterization of non-special divisors with arbitrary degree and support possibly containing non-totally ramified places.
Using this characterization and pure gaps, we explicitly constructed non-special divisors of degree $g-1$ on the GK curve and three families of effective non-special divisors of degree $g$ on Kummer extensions with the same multiplicities.
We then developed a general framework for constructing LCPs of AG codes and showed how to obtain LCD AG codes under specific conditions. Finally, we illustrated our results with explicit examples.

While Theorem \ref{theorem:non-special_thm} gives a characterization of non-special divisors, the existence problem for these divisors is left as an open question.
In Section \ref{section5}, we fix $I = \varnothing$ and $\omega = dx$ when constructing LCPs of AG codes and LCD AG codes from Corollary \ref{theorem:LCP_and_LCD}.
In fact, more new families of LCPs of AG codes and LCD AG codes can be generated by varying the choice of $I$ and the Weil differential $\omega$.
In the future, we will focus on the existence of non-special divisors characterized in Theorem \ref{theorem:non-special_thm} and exploit alternative parameter selections to construct further families of LCPs of codes and LCD codes.

\section*{Acknowledgement}
This work is supported by the National Natural Science Foundation of China (No. 12441107),  Guangdong Basic and Applied Basic Research Foundation (No. \seqsplit{2025A1515011764}), and the National Key Research and Development Program of China (No.~\seqsplit{2025YFA1017100}). 

\bibliographystyle{ieeetr}
\bibliography{ref}

@book{stichtenothAlgebraicFunctionFields2009,
  title = {Algebraic {{Function Fields}} and {{Codes}}},
  author = {Stichtenoth, Henning},
  year = {2009},
  series = {Graduate {{Texts}} in {{Mathematics}}},
  number = {254},
  publisher = {Springer Berlin Heidelberg},
  doi = {10.1007/978-3-540-76878-4},
  isbn = {978-3-540-76877-7 978-3-540-76878-4},
  edition = {Second},
  address = {Berlin}
}

@book{grahamConcreteMathematics1994,
  title = {{Concrete Mathematics: A Foundation} for {Computer Science}},
  shorttitle = {Concrete Mathematics},
  author = {Graham, Ronald L. and Knuth, Donald Ervin and Patashnik, Oren},
  year = 1994,
  edition = {2nd},
  publisher = {Addison-Wesley},
  address = {Reading, Mass},
  isbn = {978-0-201-55802-9},
  lccn = {QA39.2 .G733 1994}
}

@article{huMultipointCodesKummer2018,
  title = {Multi-Point Codes over {{Kummer}} Extensions},
  author = {Hu, Chuangqiang and Yang, Shudi},
  year = 2018,
  journal = {Designs, Codes and Cryptography},
  volume = {86},
  number = {1},
  pages = {211--230},
  issn = {0925-1022, 1573-7586},
  doi = {10.1007/s10623-017-0335-7}
}

@article{bhowmickLinearComplementary2024,
  title = {On Linear Complementary Pairs of Algebraic Geometry Codes over Finite Fields},
  author = {Bhowmick, Sanjit and Dalai, Deepak Kumar and Mesnager, Sihem},
  year = 2024,
  journal = {Discrete Mathematics},
  volume = {347},
  number = {12},
  pages = {114193},
  issn = {0012365X},
  doi = {10.1016/j.disc.2024.114193}
}

@article{mesnagerComplementaryDual2018,
  title = {Complementary {{Dual Algebraic Geometry Codes}}},
  author = {Mesnager, Sihem and Tang, Chunming and Qi, Yanfeng},
  year = 2018,
  journal = {IEEE Transactions on Information Theory},
  volume = {64},
  number = {4},
  pages = {2390--2397},
  issn = {0018-9448, 1557-9654},
  doi = {10.1109/TIT.2017.2766075}
}

@article{balletExistenceNonspecial2006,
  title = {On the Existence of Non-Special Divisors of Degree g and g - 1 in Algebraic Function Fields over $\mathbb{F}_q$},
  author = {Ballet, S. and Le Brigand, D.},
  year = 2006,
  journal = {Journal of Number Theory},
  volume = {116},
  number = {2},
  pages = {293--310},
  issn = {0022314X},
  doi = {10.1016/j.jnt.2005.04.009}
}

@article{giuliettiNewFamily2009,
  title = {A New Family of Maximal Curves over a Finite Field},
  author = {Giulietti, Massimo and Korchm{\'a}ros, G{\'a}bor},
  year = 2009,
  journal = {Mathematische Annalen},
  volume = {343},
  number = {1},
  pages = {229--245},
  issn = {0025-5831, 1432-1807},
  doi = {10.1007/s00208-008-0270-z}
}

@article{morenoExplicitNonspecial2024,
  title = {Explicit {{Non-special Divisors}} of {{Small Degree}}, {{Algebraic Geometric Hulls}}, and {{LCD Codes}} from {{Kummer Extensions}}},
  author = {Moreno, Eduardo Camps and L{\'o}pez, Hiram H. and Matthews, Gretchen L.},
  year = 2024,
  journal = {SIAM Journal on Applied Algebra and Geometry},
  volume = {8},
  number = {2},
  pages = {394--413},
  issn = {2470-6566},
  doi = {10.1137/21M1467936}
}

@article{castellanosOneTwoPoint2016,
  title = {One- and {{Two-Point Codes Over Kummer Extensions}}},
  author = {Castellanos, Alonso S. and Masuda, Ariane M. and Quoos, Luciane},
  year = 2016,
  journal = {IEEE Transactions on Information Theory},
  volume = {62},
  number = {9},
  pages = {4867--4872},
  issn = {0018-9448, 1557-9654},
  doi = {10.1109/TIT.2016.2583437}
}

@article{masseyLinearCodes1992,
  title = {Linear Codes with Complementary Duals},
  author = {Massey, James L.},
  year = 1992,
  journal = {Discrete Mathematics},
  volume = {106--107},
  pages = {337--342},
  issn = {0012365X},
  doi = {10.1016/0012-365X(92)90563-U}
}

@article{yangConditionCyclic1994,
  title = {The Condition for a Cyclic Code to Have a Complementary Dual},
  author = {Yang, Xiang and Massey, James L.},
  year = 1994,
  journal = {Discrete Mathematics},
  volume = {126},
  number = {1-3},
  pages = {391--393},
  issn = {0012365X},
  doi = {10.1016/0012-365X(94)90283-6}
}

@incollection{bringerOrthogonalDirect2014,
  title = {Orthogonal {{Direct Sum Masking}}: {{A Smartcard Friendly Computation Paradigm}} in a {{Code}}, with {{Builtin Protection}} against {{Side-Channel}} and {{Fault Attacks}}},
  shorttitle = {Orthogonal {{Direct Sum Masking}}},
  booktitle = {Information {{Security Theory}} and {{Practice}}. {{Securing}} the {{Internet}} of {{Things}}},
  author = {Bringer, Julien and Carlet, Claude and Chabanne, Herv{\'e} and Guilley, Sylvain and Maghrebi, Houssem},
  editor = {Hutchison, David and Kanade, Takeo and Kittler, Josef and Kleinberg, Jon M. and Kobsa, Alfred and Mattern, Friedemann and Mitchell, John C. and Naor, Moni and Nierstrasz, Oscar and Pandu Rangan, C. and Steffen, Bernhard and Terzopoulos, Demetri and Tygar, Doug and Weikum, Gerhard and Naccache, David and Sauveron, Damien},
  year = 2014,
  volume = {8501},
  pages = {40--56},
  publisher = {Springer Berlin Heidelberg},
  address = {Berlin, Heidelberg},
  doi = {10.1007/978-3-662-43826-8_4},
  isbn = {978-3-662-43825-1 978-3-662-43826-8}
}

@incollection{carletComplementaryDual2015,
  title = {Complementary {{Dual Codes}} for {{Counter-Measures}} to {{Side-Channel Attacks}}},
  booktitle = {Coding {{Theory}} and {{Applications}}},
  author = {Carlet, Claude and Guilley, Sylvain},
  editor = {Pinto, Raquel and Rocha Malonek, Paula and Vettori, Paolo},
  year = 2015,
  volume = {3},
  pages = {97--105},
  publisher = {Springer International Publishing},
  address = {Cham},
  doi = {10.1007/978-3-319-17296-5_9},
  isbn = {978-3-319-17295-8 978-3-319-17296-5}
}

@inproceedings{ngoLinearComplementary2015,
  title = {Linear Complementary Dual Code Improvement to Strengthen Encoded Circuit against Hardware {{Trojan}} Horses},
  booktitle = {2015 {{IEEE International Symposium}} on {{Hardware Oriented Security}} and {{Trust}} ({{HOST}})},
  author = {Ngo, Xuan Thuy and Bhasin, Shivam and Danger, Jean-Luc and Guilley, Sylvain and Najm, Zakaria},
  year = 2015,
  pages = {82--87},
  publisher = {IEEE},
  address = {Washington, DC},
  doi = {10.1109/HST.2015.7140242},
  isbn = {978-1-4673-7421-7}
}

@inproceedings{ngoEncodingState2014,
  title = {Encoding the State of Integrated Circuits: A Proactive and Reactive Protection against Hardware {{Trojans}} Horses},
  shorttitle = {Encoding the State of Integrated Circuits},
  booktitle = {Proceedings of the 9th {{Workshop}} on {{Embedded Systems Security}}},
  author = {Ngo, Xuan Thuy and Guilley, Sylvain and Bhasin, Shivam and Danger, Jean-Luc and Najm, Zakaria},
  year = 2014,
  pages = {1--10},
  publisher = {ACM},
  address = {New Delhi India},
  doi = {10.1145/2668322.2668329},
  isbn = {978-1-4503-2932-3}
}

@article{carletLinearCodes2018,
  title = {Linear {{Codes Over}} $\mathbb {F}_q$ {{Are Equivalent}} to {{LCD Codes}} for $q>3$},
  author = {Carlet, Claude and Mesnager, Sihem and Tang, Chunming and Qi, Yanfeng and Pellikaan, Ruud},
  year = 2018,
  journal = {IEEE Transactions on Information Theory},
  volume = {64},
  number = {4},
  pages = {3010--3017},
  issn = {0018-9448, 1557-9654},
  doi = {10.1109/TIT.2018.2789347}
}

@article{jinAlgebraicGeometry2018,
  title = {Algebraic {{Geometry Codes With Complementary Duals Exceed}} the {{Asymptotic Gilbert-Varshamov Bound}}},
  author = {Jin, Lingfei and Xing, Chaoping},
  year = 2018,
  journal = {IEEE Transactions on Information Theory},
  volume = {64},
  number = {9},
  pages = {6277--6282},
  issn = {0018-9448, 1557-9654},
  doi = {10.1109/TIT.2017.2773057}
}

@article{carletLinearComplementary2018,
  title = {On {{Linear Complementary Pairs}} of {{Codes}}},
  author = {Carlet, Claude and Guneri, Cem and Ozbudak, Ferruh and Ozkaya, Buket and Sole, Patrick},
  year = 2018,
  journal = {IEEE Transactions on Information Theory},
  volume = {64},
  number = {10},
  pages = {6583--6589},
  issn = {0018-9448, 1557-9654},
  doi = {10.1109/TIT.2018.2796125}
}

@article{guneriLinearComplementary2018,
  title = {On {{Linear Complementary Pair}} of  $n$ {{D Cyclic Codes}}},
  author={Güneri, Cem and Özkaya, Buket and Sayıcı, Selcen},
  year = 2018,
  journal = {IEEE Communications Letters},
  volume = {22},
  number = {12},
  pages = {2404--2406},
  issn = {1089-7798, 1558-2558, 2373-7891},
  doi = {10.1109/LCOMM.2018.2872046}
}

@article{sendrierLinearCodes2004,
  title = {Linear Codes with Complementary Duals Meet the {{Gilbert}}--{{Varshamov}} Bound},
  author = {Sendrier, Nicolas},
  year = 2004,
  journal = {Discrete Mathematics},
  volume = {285},
  number = {1-3},
  pages = {345--347},
  issn = {0012365X},
  doi = {10.1016/j.disc.2004.05.005}
}

@article{carletEuclideanHermitian2018,
  title = {Euclidean and {{Hermitian LCD MDS}} Codes},
  author = {Carlet, Claude and Mesnager, Sihem and Tang, Chunming and Qi, Yanfeng},
  year = 2018,
  journal = {Designs, Codes and Cryptography},
  volume = {86},
  number = {11},
  pages = {2605--2618},
  issn = {0925-1022, 1573-7586},
  doi = {10.1007/s10623-018-0463-8}
}

@article{sokConstructionsOptimal2018,
  title = {Constructions of Optimal {{LCD}} Codes over Large Finite Fields},
  author = {Sok, Lin and Shi, Minjia and Sol{\'e}, Patrick},
  year = 2018,
  journal = {Finite Fields and Their Applications},
  volume = {50},
  pages = {138--153},
  issn = {10715797},
  doi = {10.1016/j.ffa.2017.11.007}
}

@article{beelenExplicitMDS2018,
  title = {Explicit {{MDS Codes With Complementary Duals}}},
  author = {Beelen, Peter and Jin, Lingfei},
  year = 2018,
  journal = {IEEE Transactions on Information Theory},
  volume = {64},
  number = {11},
  pages = {7188--7193},
  issn = {0018-9448, 1557-9654},
  doi = {10.1109/TIT.2018.2816934}
}

@article{bhowmickLinearComplementary2024a,
  title = {Linear Complementary Pairs of Codes over a Finite Non-Commutative {{Frobenius}} Ring},
  author = {Bhowmick, Sanjit and Liu, Xiusheng},
  year = 2024,
  journal = {Journal of Applied Mathematics and Computing},
  volume = {70},
  number = {5},
  pages = {4923--4936},
  issn = {1598-5865, 1865-2085},
  doi = {10.1007/s12190-024-02161-w}
}

@article{liSeveralConstructions2024,
  title = {Several Constructions of Optimal {{LCD}} Codes over Small Finite Fields},
  author = {Li, Shitao and Shi, Minjia and Liu, Huizhou},
  year = 2024,
  journal = {Cryptography and Communications},
  volume = {16},
  number = {4},
  pages = {779--800},
  issn = {1936-2447, 1936-2455},
  doi = {10.1007/s12095-024-00699-x}
}

@article{borelloNoteLinear2020,
  title = {A Note on Linear Complementary Pairs of Group Codes},
  author = {Borello, Martino and De La Cruz, Javier and Willems, Wolfgang},
  year = 2020,
  journal = {Discrete Mathematics},
  volume = {343},
  number = {8},
  pages = {111905},
  issn = {0012365X},
  doi = {10.1016/j.disc.2020.111905}
}

@article{huLinearComplementary2021,
  title = {Linear Complementary Pairs of Codes over Rings},
  author = {Hu, Peng and Liu, Xiusheng},
  year = 2021,
  journal = {Designs, Codes and Cryptography},
  volume = {89},
  number = {11},
  pages = {2495--2509},
  issn = {0925-1022, 1573-7586},
  doi = {10.1007/s10623-021-00933-0}
}

@article{ishizukaConstructionBoth2023,
  title = {Construction for Both Self-Dual Codes and {{LCD}} Codes},
  author = {Ishizuka, Keita and Saito, Ken and {Research Center for Pure and Applied Mathematics, Graduate School of Information Sciences, Tohoku University, Sendai 980-8579, Japan}},
  year = 2023,
  journal = {Advances in Mathematics of Communications},
  volume = {17},
  number = {1},
  pages = {139--151},
  issn = {1930-5346, 1930-5338},
  doi = {10.3934/amc.2021070}
}

@article{haradaConstructionBinary2021,
  title = {Construction of Binary {{LCD}} Codes, Ternary {{LCD}} Codes and Quaternary {{Hermitian LCD}} Codes},
  author = {Harada, Masaaki},
  year = 2021,
  journal = {Designs, Codes and Cryptography},
  volume = {89},
  number = {10},
  pages = {2295--2312},
  issn = {0925-1022, 1573-7586},
  doi = {10.1007/s10623-021-00916-1}
}

@article{goppaALGEBRAICOGEOMETRICCODES1983,
  title = {{{ALGEBRAICO-GEOMETRIC CODES}}},
  author = {Goppa, V D},
  year = 1983,
  journal = {Mathematics of the USSR-Izvestiya},
  volume = {21},
  number = {1},
  pages = {75--91},
  issn = {0025-5726},
  doi = {10.1070/IM1983v021n01ABEH001641}
}

@article{castellanosLinearComplementary2025,
  title = {Linear {{Complementary Dual Codes}} and {{Linear Complementary Pairs}} of {{AG Codes}} in {{Function Fields}}},
  author = {Castellanos, Alonso S. and Marques, Adler V. and Quoos, Luciane},
  year = 2025,
  journal = {IEEE Transactions on Information Theory},
  volume = {71},
  number = {3},
  pages = {1676--1688},
  issn = {0018-9448, 1557-9654},
  doi = {10.1109/TIT.2024.3521094}
}

@article{maharajCodeConstruction2004,
  title = {Code {{Construction}} on {{Fiber Products}} of {{Kummer Covers}}},
  author = {Maharaj, H.},
  year = 2004,
  journal = {IEEE Transactions on Information Theory},
  volume = {50},
  number = {9},
  pages = {2169--2173},
  issn = {0018-9448},
  doi = {10.1109/TIT.2004.833356}
}

@article{kimIndexWeierstrass1994,
  title = {On the Index of the {{Weierstrass}} Semigroup of a Pair of Points on a Curve},
  author = {Kim, Seon Jeong},
  year = 1994,
  journal = {Archiv der Mathematik},
  volume = {62},
  number = {1},
  pages = {73--82},
  issn = {0003-889X, 1420-8938},
  doi = {10.1007/BF01200442}
}

@incollection{matthewsWeierstrassSemigroup2004,
  title = {The {{Weierstrass Semigroup}} of an M-Tuple of {{Collinear Points}} on a {{Hermitian Curve}}},
  booktitle = {Finite {{Fields}} and {{Applications}}},
  author = {Matthews, Gretchen L.},
  editor = {Goos, Gerhard and Hartmanis, Juris and Van Leeuwen, Jan and Mullen, Gary L. and Poli, Alain and Stichtenoth, Henning},
  year = 2004,
  volume = {2948},
  pages = {12--24},
  publisher = {Springer Berlin Heidelberg},
  address = {Berlin, Heidelberg},
  doi = {10.1007/978-3-540-24633-6_2},
  isbn = {978-3-540-21324-6 978-3-540-24633-6}
}

@article{carvalhoGoppaCodes2005,
  title = {On {{Goppa Codes}} and {{Weierstrass Gaps}} at {{Several Points}}},
  author = {Carvalho, C{\'i}cero and Torres, Fernando},
  year = 2005,
  journal = {Designs, Codes and Cryptography},
  volume = {35},
  number = {2},
  pages = {211--225},
  issn = {0925-1022, 1573-7586},
  doi = {10.1007/s10623-005-6403-4}
}

@article{castellanosCompleteSet2024,
  title = {Complete Set of Pure Gaps in Function Fields},
  author = {Castellanos, Alonso S. and Mendoza, Erik A.R. and Tizziotti, Guilherme},
  year = 2024,
  journal = {Journal of Pure and Applied Algebra},
  volume = {228},
  number = {4},
  pages = {107513},
  issn = {00224049},
  doi = {10.1016/j.jpaa.2023.107513}
}

@article{castellanosGeneralizedWeierstrass2026,
  title = {On Generalized {{Weierstrass}} Semigroups in Arbitrary {{Kummer}} Extensions of $\mathbb{F}_q(x)$},
  author = {Castellanos, Alonso S. and Mendoza, Erik and Tizziotti, Guilherme},
  year = 2026,
  journal = {Finite Fields and Their Applications},
  volume = {112},
  pages = {102808},
  issn = {10715797},
  doi = {10.1016/j.ffa.2026.102808}
}

@article{zhangPureGaps2026,
  title = {Pure Gaps at Many Places and Multi-Point {{AG}} Codes from Arbitrary {{Kummer}} Extensions},
  author = {Zhang, Huachao and Zhao, Chang-An},
  year = 2026,
  journal = {Designs, Codes and Cryptography},
  volume = {94},
  number = {6},
  pages = {130},
  issn = {0925-1022, 1573-7586},
  doi = {10.1007/s10623-026-01866-2}
}

@misc{junjieLinearComplementary2025,
  title = {Linear {{Complementary Pairs}} of {{Algebraic Geometry Codes}} via {{Kummer Extensions}}},
  author = {Junjie, Huang and Haojie, Chen and Huachao, Zhang and {Chang-An}, Zhao},
  year = 2025,
  note = {arXiv:2506.23081},
  doi = {10.48550/arXiv.2506.23081},
}

@misc{mendozaCharacterizationNonspecial2026,
  title = {Characterization of Non-Special Divisors of Small Degree on {{Kummer}} Extensions and {{LCP}} Codes},
  author = {Mendoza, Erik and Navarro, Horacio and Quoos, Luciane},
  year = 2026,
  note = {	arXiv:2604.27146},
  doi = {10.48550/ARXIV.2604.27146},
}

@misc{marquesConstructionNonspecial2026,
  title = {Construction of {{Non-special Divisors}} on {{Kummer Covers With Arbritary Ramification For LCP Codes}}},
  author = {Marques, Adler and {da Silva}, Yuri and Tafazolian, Saeed},
  year = 2026,
  note = {arXiv:2605.14046},
  doi = {10.48550/ARXIV.2605.14046},
}

@misc{zhangWeierstrassSemigroups2026,
  title = {Weierstrass Semigroups at Totally Ramified Places of Degree One on {{Kummer}} Extensions},
  author = {Zhang, Huachao and Zhao, Chang-An},
  year = 2026,
  note = {arXiv:2605.14583},
  doi = {10.48550/ARXIV.2605.14583},
}
\end{document}